\documentclass{article}

\usepackage{authblk}
\usepackage{inputenc}
\usepackage{amsfonts,amsmath,amsthm,amssymb,mathtools}
\usepackage{graphics}
\usepackage{stmaryrd}
\usepackage[ruled,vlined,boxed]{algorithm2e}
\usepackage{todonotes}

\usepackage{a4wide}

\usepackage{tikz}
\usetikzlibrary{arrows, shapes, positioning, decorations.pathmorphing}

\usepackage{hyperref}

\newtheorem{thm}{Theorem}
\newtheorem{lemm}[thm]{Lemma}
\newtheorem*{lemm*}{Lemma}
\newtheorem{prop}[thm]{Proposition}
\newtheorem{coro}[thm]{Corollary}

\theoremstyle{definition}
\newtheorem{defi}[thm]{Definition}
\newtheorem{exam}[thm]{Example}
\newtheorem{rema}[thm]{Remark}

\newcommand{\Ckl}{ \mathcal{C}_{k,t}^l }
\newcommand{\Ckmupl}{ \mathcal{C}_{k-1,t}^{l,p} }
\newcommand{\Ckmul}{ \mathcal{C}_{k-1,t}^{l} }

\author[1]{Bérénice Delcroix-Oger \thanks{The work of the first author has been partially funded by ANR CARPLO.  }}
\author[1]{Matthieu Josuat-Verg\`{e}s \thanks{The work of the second author has been partially funded by ANR COMBINÉ.}}
\author[2]{Lucas Randazzo}
 
\affil[1]{Université de Paris, CNRS, IRIF, F-75006, Paris, France} 
\affil[2]{Université Gustave Eiffel, CNRS, LIGM, Champs-sur-Marne, FRANCE}

\newcommand{\pp}{\mathrlap{\,\scriptstyle{\overset{2}{\null}}}\Pi}

\newcommand{\ppdel}{\mathrlap{\scriptstyle{\overset{\hspace{1pt}\Delta}{\null}}}\Pi}

\newcommand{\ppp}{\mathrlap{\hspace{0.5mm}\scalebox{0.56}{$\overset{2}{\null}$}}\Pi}

\DeclareMathOperator{\rk}{rk}
\DeclareMathOperator{\im}{im}

\begin{document}
\setcounter{tocdepth}{2}

\title{Some properties of the parking function poset}

\maketitle

\abstract{In 1980, Edelman defined a poset on objects called the noncrossing 2-partitions.  They are closely related with noncrossing partitions and parking functions.  To some extent, his definition is a precursor of the parking space theory, in the framework of finite reflection groups.  We present some enumerative and topological properties of this poset.  In particular, we get a formula counting certain chains, that encompasses formulas for Whitney numbers (of both kinds).  We prove shellability of the poset, and compute its homology as a representation of the symmetric group. We moreover link it with two well-known polytopes : the associahedron and the permutohedron.}

\tableofcontents

\section*{Introduction}

{\it Parking functions} are fundamental objects in algebraic combinatorics.  It is well known that the set of parking functions of length $n$ has cardinality $(n+1)^{n-1}$, and the natural action of the symmetric group $\mathfrak{S}_n$ on this set occurs in the deep work of Haiman~\cite{haiman} about diagonal coinvariants.  Generalizations to other finite reflection groups lead to the {\it parking space theory} of Armstrong, Reiner, Rhoades~\cite{ARR,rhoades}.

The poset mentioned in the title was introduced by Edelman~\cite{edelman} in 1980, as a variant of the {\it noncrossing partition lattice} introduced by Kreweras~\cite{kreweras} (hence the name {\it noncrossing 2-partitions} in~\cite{edelman}).  One striking feature of Edelman's definition is that it really fits well in the noncrossing parking space theory mentionned above, so it seems that this overlooked poset can give a new perspective on recent results about parking functions.

Our goal is to obtain new enumerative and topological properties of Edelman's poset.  Through various bijections, we will see that several variants of the same objects are relevant: 
\begin{itemize}
\item 2-noncrossing partitions (Section~\ref{sec11}), 
\item some pairs of a noncrossing partition together with a permutation (Section~\ref{sec_parkingspace}),
\item parking functions in the usual way (Section~\ref{sec13}),
\item parking trees (Section~\ref{sec13}).
\end{itemize}
The latter, which have the additional structure of a {\it species}, are defined on a set $V$ as trees whose nodes are labelled by (possibly empty) parts of a weak partition of $V$ and such that a node labelled by $N$ as exactly $N$ (possibly empty) children. They will be useful to write functional equations and get our enumerative results in Section~\ref{sec_enum}.  We draw below the different representations of the same element.

\begin{figure}[h!]
    \centering
\begin{tikzpicture}
   \node(a) at (0,0) {$\{\{1,5,6,8\},$};
   \node(b) at (1.5,0) {$\{2,4\},$};
   \node(c) at (2.5,0) {$\{3\},$};
   \node(d) at (3.3,0) {$\{7\},$};
   \node(e) at (4.5,0) {$\{9,10,12\},$};
   \node(f) at (6,0) {$\{11\}\}$};
   \node(a2) at (0,-2) {$\{\{1\},$};
   \node(b2) at (1.5,-2) {$\{2,9,10,11\},$};
   \node(c2) at (3.2,-2) {$\{3,4,8\},$};
   \node(d2) at (4.2,-2) {$\{5\},$};
   \node(e2) at (5.3,-2) {$\{6,12\},$};
   \node(f2) at (6.3,-2) {$\{7\}\}$};
   \node(z) at (-1,-1) {$\lambda$};
\draw[->] (a.south)--(b2.north);
\draw[->] (b.south)--(e2.north);
\draw[->] (c.south)--(d2.north);
\draw[->] (d.south)--(f2.north);
\draw[->] (e.south)--(c2.north);
\draw[->] (f.south)--(a2.north);
\end{tikzpicture}    
 \begin{tikzpicture}[grow=up, scale=1, level distance=8mm, sibling distance=8mm]
    \tikzstyle{som} = [ellipse, draw, inner sep=1mm]
\node[som]{2 \ 9  \ 10 \ 11}
     child{
            node[som]{3 \ 4 \ 8}
            child{
                node{}
            }
            child{
               node[som]{1}
               child{
                    node{}
               }
            }
            child{
                node{}
            }
        }        
    child{
        node{}
    }
    child { 
        node[som] {7}
        child{
            node{} 
        }
    }
    child{
            node{} 
    }        
    child{
            node[som]{6 \ 12}
            child{
                node{}
            }
            child{
                node[som]{5}
                child{}
            }
    }        
;
\end{tikzpicture}
     \begin{tikzpicture}[scale=0.6]
   \tikzstyle{ver} = [circle, draw, fill, inner sep=0.5mm]
   \tikzstyle{edg} = [line width=0.6mm]
   \node[ver] at (1,0) {};
   \node[ver] at (2,0) {};
   \node[ver] at (3,0) {};
   \node[ver] at (4,0) {};
   \node[ver] at (5,0) {};
   \node[ver] at (6,0) {};
   \node[ver] at (7,0) {};
   \node[ver] at (8,0) {};
   \node[ver] at (9,0) {};
   \node[ver] at (10,0) {};
   \node[ver] at (11,0) {};
   \node[ver] at (12,0) {};
   \node      at (1,-0.6) {2};
   \node      at (2,-0.6) {6};
   \node      at (3,-0.6) {5};
   \node      at (4,-0.6) {12};
   \node      at (5,-0.6) {9};
   \node      at (6,-0.6) {10};
   \node      at (7,-0.6) {7};
   \node      at (8,-0.6) {11};
   \node      at (9,-0.6) {3};
   \node      at (10,-0.6) {4};
   \node      at (11,-0.6) {1};
   \node      at (12,-0.6) {8};
   \draw[edg] (1,0) to[bend left=60] (5,0);
   \draw[edg] (5,0) to[bend left=60] (6,0);
   \draw[edg] (2,0) to[bend left=60] (4,0);
   \draw[edg] (6,0) to[bend left=60] (8,0);
   \draw[edg] (9,0) to[bend left=60] (10,0);
   \draw[edg] (10,0) to[bend left=60] (12,0);
 \end{tikzpicture}
     $12 \ 1 \ 10\ 10\ 3\ 2\ 7\ 10\ 1\ 1\ 1\ 2$
 \caption{Four representations of the same element : as a 2-noncrossing partitions, as a parking tree, as a pair of a noncrossing partition together with a permutation and as a parking function}
    \label{fig:4park}
\end{figure}
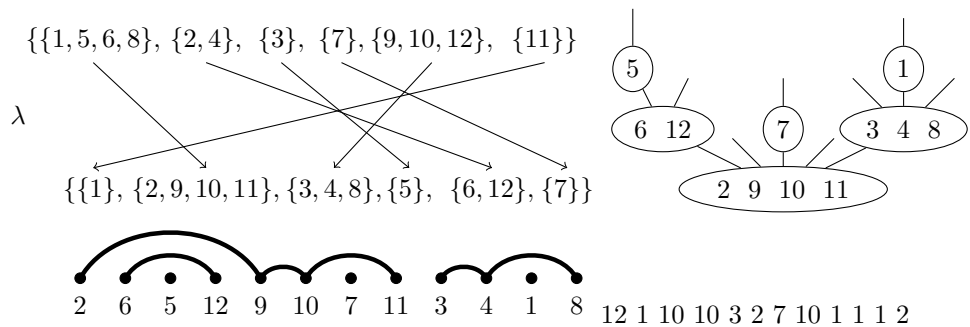

What we get is the following formula counting chains of $k$ elements whose top element has rank $\ell$:
\begin{equation*} 
    \ell !  \binom{kn}{\ell} S_2(n,\ell+1).
 \end{equation*}
 A nice feature of this formula is that it encompasses a nice formula for Whitney numbers of the second kind at $k=1$ (this one being obtained by Edelman), and one for Whitney numbers of the first kind at $k=-1$. We give a bijective proof of this formula in Section~\ref{ktree}, introducing $k$-parking trees, which generalise parking trees and encode chains of $k$ elements in the parking posets.

Then we go on to topological properties: we will see in Section~\ref{sec_shell} that the poset is shellable.  Unlike the case of noncrossing partitions which can be treated by EL-shellability, we need here an {\it ad hoc} property of the lattice. More precisely, we find an order on atoms of every intervals and we show that this order induces the shellability of the poset as it satisfies Lemma~\ref{lemma_for_shelling}, which we recall here:
\begin{lemm*} 
  Let $x,y,y',z\in\pp_n$ such that $x\lessdot y \lessdot z$, $x\lessdot y'$, and $y'\prec_x y$.  Then:
  \begin{itemize}
      \item either there exists $y'' \in \pp_n$ such that $x\lessdot y'' \lessdot z$ and $y''\prec_x y$,
      \item or there exists $z'\in\pp_n$ such that $y\lessdot z' \leq y'\vee z$ and $z' \prec_y z$.
  \end{itemize}
\end{lemm*}
 This lemma gives a new criterion on orders on atoms of a poset to prove shellability. 
Still, the EL-shellability of noncrossing partitions is a key tool.  There are well known consequences of shellability such as Cohen-Macaulayness, and hence that only one homology group of the poset is non trivial.  We use this fact in Section~\ref{sec_homology} to compute the character of this homology group as a representation of $\mathfrak{S}_n$.

This poset is finally also deeply linked with the permutahedron (Section~\ref{permutohedron}), as it contains its face poset as a subposet,  and with the associahedron (Section~\ref{associahedron}) through cluster parking functions.

\section{Parking function posets}
\label{sec1}

\label{sec11}
\subsection{Set partitions and noncrossing partitions}

\label{sec:nc}

For integers $i\leq j$, we will denote $\llbracket i;j\rrbracket := \{i,i+1 ,\dots, j\}$.  Let $\Pi_n$ denote the lattice of set partitions of $\llbracket 1;n\rrbracket$, endowed with refinement order. Note that we take the convention that the minimal element is $\{\llbracket 1;n\rrbracket \}$ (the set partition with one block, denoted $0_n$), and the maximal element is $\{\{1\},\{2\},\dots,\{n\}\}$ (the set partition with $n$ blocks, denoted $1_n$).

A set partition $\pi\in\Pi_n$ is {\it noncrossing} if there exists no $i<j<k<l$ such that $i,k\in B_1$ and $j,l\in B_2$ where $B_1$ and $B_2$ are two distinct blocks of $\pi$.  Endowed with the refinement order, noncrossing partitions of $\llbracket 1 ; n\rrbracket$ form a lattice denoted $NC_n$, first defined by Kreweras~\cite{kreweras}.  It is also a (full) subposet of $\Pi_n$.  The cardinality of $NC_n$ is the $n$th Catalan number, defined as
\[
  C_n := \frac{1}{2n+1}\binom{2n+1}{n}.
\]
More generally, the zeta polynomial of $NC_n$ is given by the Fuß-Catalan number $C_n^{(k)}$:
\[
  Z(NC_n,k) = C_n^{(k)} := \frac{1}{kn+1}\binom{kn+1}{n}.
\]
(Recall that, by definition, $Z(NC_n,k)$ is the number of $k-1$-element weakly increasing sequences in $NC_n$.)
In particular, the Möbius number of $NC_n$ is $Z(NC_n,-1) = (-1)^{n-1}C_{n-1}$.  This integer $C_n^{(k)}$ is also the number of $k$-{\it trees} with $n$ internal vertices, which are by definition rooted plane trees where each of the $n$ internal vertices has $ki$ descendant for some $i\in \mathbb{N}$.  This can be proved by showing that the generating function $T(z)$ of $k$-trees satisfy
\[
  T(z) = \frac{1}{1-T(z)^k}.
\]

There exists a bijection $\beta$ between $NC_n$ and rooted plane trees with $n+1$ vertices.  It can be described inductively as follows.  For $\pi\in NC_n$, let $b\in\pi$ denote the block containing $1$.  If $b=\{i_1,\dots,i_k\}$, then the root of $\beta(\pi)$ has $k$ descendants, and the $j$th subtree 

There is a natural embedding of $NC_n$ in the symmetric group~\cite{Biane}.
To each noncrossing partition $\pi\in NC_n$, we associate a permutation $\bar\pi \in \mathfrak{S}_n$ having one cycle for each block of $\pi$: for $B=\{b_1,b_2,\dots,b_k\}\in\pi$ we have $\bar\pi(b_i)=b_{i+1}$ if $i<k$, and $\bar\pi(b_k)=b_1$.  This permits us to define the {\it Kreweras complement} \cite{kreweras}: for $\pi\in NC_n$, it is $K(\pi)\in NC_n$ such that the associated permutation is $\bar 0_n \bar\pi^{-1}$ (recall that $0_n$ is the minimal noncrossing partition, so that $\bar 0_n$ is an $n$-cycle).  The map $\pi\mapsto K(\pi)$ is an anti-automorphism of $NC_n$.  For example, $K(\{\{1,2\},\{3\},\{4,5,6\}\}) = \{\{1,3,4\},\{2\},\{5\},\{6\}\}$ since in the symmetric group we have $(123456)(12)(654)=(134)$.

To end this section, we note the following facts about order ideals and filters in $NC_n$.  It is easy to see that the order filter containing all elements above some $\pi\in NC_n$ is isomorphic to the product $NC_{i_1}\times NC_{i_2}\times \cdots$, where $i_1,i_2,\dots$ are the block sizes of $\pi$.  As the anti-automorphism $K$ sends order ideals to order filters, the order ideal containing all elements below some $\pi\in NC_n$ is isomorphic to the product $NC_{i_1}\times NC_{i_2}\times \cdots$, where $i_1,i_2,\dots$ are the block sizes of $K(\pi)$.

\subsection{Noncrossing 2-partitions}

The following definition is due to Edelman~\cite{edelman}.

\begin{defi}[\cite{edelman}]  \label{defi_edelman}
A {\it noncrossing 2-partition} of $\llbracket 1; n\rrbracket$ is a 
triple $(\pi,\rho,\lambda)$ where:
\begin{itemize}
 \item $\pi \in NC_n$ and $\rho\in \Pi_n$,
 \item $\lambda$ is a bijection from (the blocks of) $\pi$ to (those of) $\rho$, and $\forall B\in \pi$, $|\lambda(B)| = |B|$.
\end{itemize}
This set is denoted $\pp_n$.  A partial order on $\pp_n$ is defined by $(\pi,\rho,\lambda) \geq (\pi',\rho',\lambda')$ iff:
\begin{itemize}
 \item $\pi$ is a refinement of $\pi'$, $\rho$ is a refinement of $\rho'$,
 \item if $\biguplus_{i=1}^j B_i = B'$ where $B_i\in \pi$ and $B'\in\pi'$, then $\biguplus_{i=1}^j \lambda(B_i) = \lambda'(B')$.
\end{itemize}
\end{defi}

Note that there is some redundancy in the notation as above, as the blocks of $\rho$ are the images $\lambda(B)$ for $B\in\pi$.  
For example, such a triple $(\pi,\rho,\lambda)$ is as follows: $\pi= \{\{1,5,6,8\},\allowbreak \{2,3\},\allowbreak \{4\}, \allowbreak \{7\}\}$, $\rho$ and $\lambda$ are given by $\lambda(\{1,5,6,8\})=\{2,3,4,7\}$, $\lambda(\{2,3\})=\{5,8\}$, $\lambda(\{4\})=\{1\}$, $\lambda(\{7\})=\{6\}$.  Another representation will be given in Section~\ref{sec_parkingspace} (in particular, see the example at the end).

Many properties of this poset will follow from the following:

\begin{lemm}  \label{lem:unique}
  Let $(\pi,\rho,\lambda) \in \pp_n $, and let $\pi' \in NC_n$ such that $\pi' \leq \pi$.  Then, there exists unique $\rho'$ and $\lambda'$ such that $(\pi',\rho',\lambda') \in \pp_n $ and $(\pi',\rho',\lambda') \leq (\pi,\rho,\lambda)$.
\end{lemm}

\begin{proof}
  Let $B'\in \pi'$, and write $B' = \biguplus_{i=1}^j B_i$ where $B_i\in \pi$.  By the second condition in the definition of the order, the only possible choice is to define $\lambda'(B')$ as $\biguplus_{i=1}^j \lambda(B_i)$, and $\rho'$ as the set of all such $\lambda'(B')$.  This definition ensures that blocks of $\rho'$ are union of blocks of $\rho$ (so $\rho'\leq\rho$ in $\Pi_n$), and $(\pi',\rho',\lambda')$ satisfies the required properties.
\end{proof}

The following two propositions contains simple consequences of the definitions or of the previous lemma.  Proofs are straightforward.

\begin{prop} The poset $\pp_n$
  \begin{itemize}
    \item is ranked, with rank function $\rk((\pi,\rho,\lambda)) = |\pi|-1$,
    \item has one minimal element, namely $(0_n, 0_n, id)$ (in general, $id$ will be the identity map of some set which is not specified in the notation but clear from the context), 
    \item has $n!$ maximal elements, namely the elements $(1_n, 1_n, \sigma)$ for $\sigma\in \mathfrak{S}_n$.
  \end{itemize}
\end{prop}

\begin{prop} \label{lemm:intervproj}
For each $(\pi,\rho,\lambda)\in\pp_n$, the order ideal of elements below $(\pi,\rho,\lambda)$ in $\pp_n$ is isomorphic to the order ideal of elements below $\pi$ in $NC_n$, via the projection $(\pi',\rho',\lambda') \mapsto \pi'$.  
\end{prop}

In particular, for each maximal element $\phi\in\pp_n$, the interval of $\pp_n$ bounded by the minimal element and $\phi$ is isomorphic to $NC_n$. Thus, $\pp_n$ can be seen as $n!$ copies of $NC_n$, which blend in a rather nontrivial way.

From the previous proposition and the structure of order ideals in $NC_n$, we get the following:

\begin{prop}
For each $(\pi,\rho,\lambda)\in\pp_n$, the order ideal of elements below $(\pi,\rho,\lambda)$ in $\pp_n$ is isomorphic to the product $NC_{i_1}\times NC_{i_2}\times \cdots$, where $i_1,i_2,\dots$ are the block sizes of $K(\pi)$.
\end{prop}

Similarly, the following result is straightforward.  This naturally extends the remark about the structure of order filters in $NC_n$.

\begin{prop}
  The order filter of $\pp_n$ containing all elements above $(\pi,\rho,\lambda)$ is isomorphic to a product $\pp_{i_1}\times \pp_{i_2}\times \cdots$, where $i_1,i_2,\dots$ are the block sizes $\pi$.
\end{prop}

For a non-crossing $2$-partition $ \phi = (\pi,\rho,\lambda)\in\pp_n$ and $k \in \llbracket 1;n\rrbracket$, let $\eta_\phi(k)$ denote the part $B\in\pi$ such that $k\in\lambda(B)$.  Note that $\phi$ is characterized by the sequence $(\eta_\phi(k))_{1\leq k\leq n}$.  For a given $\pi\in NC_n$, sequences obtained in this way are length $n$ sequences such that $B\in\pi$ appears exactly $|B|$ times.  The following lemma is deduced directly from the definition of the $2$-noncrossing partition poset:
 
\begin{lemm}
 \label{lem:c_eta}
 Let $\phi,\psi \in \pp_n$. Then $\phi \leq \psi$ if and only if $\forall k\in\llbracket 1;n\rrbracket$,  $\eta_\psi(k) \subset \eta_\phi(k)$.
\end{lemm}

Next, we consider the poset $\hat\pp_n$ obtained from $\pp_n$ by adding a new maximal element $\hat 1$. It is thus a {\it bounded} poset (it has one minimal element and one maximal element).

\begin{prop}
 \label{prop:lattice}
 The poset $\hat\pp_n$ is a lattice.  
\end{prop}
 
\begin{proof}
As the poset is bounded, it suffices to prove that it is a join-semilattice, i.e., for $\phi,\psi\in\pp_n$ there is a least upper bound $\phi\vee\psi \in\hat\pp_n$. Indeed, the meet $\phi \wedge \psi$ is then obtained as the join of all elements that are below both $\phi$ and $\psi$.  Write $\phi=(\pi,\rho,\lambda)$ and $\psi=(\pi',\rho',\lambda')$.  Let $\gamma$ be the meet $\pi\vee\pi'$ in $NC_n$, so that the blocks of $\gamma$ are the nonempty intersection $B_1\cap B_2$ for $B_1\in\pi$ and $B_2\in\pi'$.

Using the previous lemma, it is natural to consider the sequence $(\eta_\phi(k) \cap \eta_\psi(k))_{1\leq k  \leq n}$.  Each element is either empty or a block of $\gamma$.  We distinguish several cases.

In the first case, assume that $\eta_\phi(k) \cap \eta_\psi(k) = \varnothing$ for some $k$.  It means there is no $\chi\in\pp_n$ above $\phi$ and $\psi$, as the inclusion $\eta_\chi(k) \subset \eta_\phi(k) \cap \eta_\psi(k)$ is impossible.  So $\hat 1$ is the least upper bound of $\phi$ and $\psi$ in $\hat\pp_n$.

In the second case, assume that there are $b\in\gamma$ and $\ell$ indices $k_1,\dots,k_\ell$ such that $\eta_\phi(k_i) \cap \eta_\psi(k_i) = b$ (for $1\leq i \leq \ell$), and $\ell>|b|$.  By way of contradiction, assume that there is a maximal element $\chi$ above $\phi$ and $\psi$.  Then, the sets $\eta_\chi(k_i)$ are pairwise distinct singletons (by maximality of $\chi$), and they are included in $b$ (by the previous lemma, since $\chi\geq\phi$ and $\chi\geq\psi$). This is not possible as $\ell>|b|$, and it follows that $\hat 1$ is the least upper bound of $\phi$ and $\psi$.

In the remaining cases, we show that there exists $\chi\in\pp_n$ such that $\eta_\chi(k) = \eta_\phi(k) \cap \eta_\psi(k)$ for $1\leq k\leq n$.  Indeed, each $b\in\gamma$ appears at most $|b|$ times in the sequence $(\eta_\phi(k) \cap \eta_\psi(k))_{1\leq k  \leq n}$.  Since the elements in the sequence are nonempty, a counting argument shows that each $b\in\pi$ appears exactly $|b|$ times, as needed.  By the previous lemma, $\chi$ is a least upper bound of $\phi$ and $\psi$.
\end{proof}

\subsection{The parking space}
\label{sec_parkingspace}

The goal of this section is to describe a natural action of the symmetric group $\mathfrak{S}_n$ on $\pp_n$.  It leads to a connection with {\it parking spaces} as defined by Armstrong, Reiner and Rhoades in \cite{ARR}, and to an alternative definition of $2$-noncrossing partitions as some pairs $(\pi,\sigma)\in NC_n\times \mathfrak{S}_n$.

The natural action of $\mathfrak{S}_n$ on $\llbracket 1;n\rrbracket$ can be extended to various other combinatorial sets by the rule $\sigma\cdot X = \{\sigma \cdot x \; : \; x\in X \}$.  This extended action automatically respects properties of sets such as inclusion, disjoint unions, etc.  In particular, this gives a natural action of $\mathfrak{S}_n$ on $\Pi_n$ which respects the poset structure.  

\begin{prop}
  There is an action of $\mathfrak{S}_n$ on $\pp_n$ defined by 
  \begin{equation}  \label{act_ontriples}
    \sigma \cdot (\pi,\rho,\lambda) = (\pi, \sigma\cdot\rho , \sigma \circ \lambda ),
  \end{equation}
  where in $\sigma \circ \lambda$ we identify $\sigma$ with its action on set partitions.  This action preserves the order relation of $\pp_n$, so that it can be extended to an action on the chains of $\pp_n$.
\end{prop}

\begin{proof}
  This is a direct application of the properties of the action of $\mathfrak{S}_n$ on $\Pi_n$.
\end{proof}

Note that $(\pi,\rho,\lambda)$ and $(\pi',\rho',\lambda')$ are in the same orbit if and only if $\pi=\pi'$, so that the orbits are naturally indexed by $NC_n$.

For $\pi\in NC_n$, we denote by $\mathfrak{S}_n(\pi)$ the set of $\sigma\in\mathfrak{S}_n$ such that $\sigma\cdot  b = b$ for any block $b\in\pi$. Then $\mathfrak{S}_n(\pi)$ is a {\it parabolic subgroup} (it is conjugated to a Young subgroup).  Via left multiplication, the quotient $\mathfrak{S}_n / \mathfrak{S}_n(\pi)$ is acted on by $\mathfrak{S}_n$.

\begin{prop}  \label{prop:equivariantbij}
  There is a $\mathfrak{S}_n$-equivariant bijection between $\pp_n$ and pairs $(\pi,\sigma\cdot \mathfrak{S}_n(\pi))$ where $\pi\in NC_n$ and $\sigma \in \mathfrak{S}_n$ (where it is understood that $\mathfrak{S}_n$ acts on the second element in the pair via left multiplication on the cosets).
\end{prop}

\begin{proof}
  To a pair $(\pi,\sigma \cdot \mathfrak{S}_n(\pi))$ as in the proposition, we associate $(\pi,\rho,\lambda)\in\pp_n$ by letting $\rho=\sigma\cdot\pi$, and $\lambda$ sends $b\in \pi$ to $\sigma\cdot b \in \rho$.  By replacing $\sigma$ with the product $\sigma_1\sigma_2$ in the definition of this map, we see that $(\pi,\sigma) \mapsto (\pi,\rho,\lambda)$ is equivariant.

  To define the inverse map, let $(\pi,\rho,\lambda)\in\pp_n$.  There exists $\sigma \in\mathfrak{S}_n$ such that $\sigma \cdot b = \lambda(b)$ for $b\in\pi$, moreover it is unique up to right multiplication by an element of $\mathfrak{S}_n(\pi)$.  The inverse bijection sends $(\pi,\rho,\lambda)\in\pp_n$ to $(\pi,\sigma)$.
\end{proof}

By taking minimal length coset representatives in each coset, we immediately obtain the following.

\begin{prop} \label{prop_pppairs}
  There is a bijection between $\pp_n$ and the set $\mathcal{P}_n$ of pairs $(\pi,\sigma) \in NC_n \times \mathfrak{S}_n$ such that if  $b_1<b_2<\ldots$ are the elements of  a block $b\in\pi$ then $\sigma(b_1)<\sigma(b_2)<\dots$.
\end{prop}

The example given after Definition~\ref{defi_edelman} gives the pair
$ (\{\{1,5,6,8\},\{2,3\},\{4\},\{7\}\}, 2 5 8 1 3 4 6 7)$.
It can be naturally represented as a noncrossing partition with labels:
\begin{align} \label{pfp_example}
 \begin{tikzpicture}[scale=0.6]
   \tikzstyle{ver} = [circle, draw, fill, inner sep=0.5mm]
   \tikzstyle{edg} = [line width=0.6mm]
   \node[ver] at (1,0) {};
   \node[ver] at (2,0) {};
   \node[ver] at (3,0) {};
   \node[ver] at (4,0) {};
   \node[ver] at (5,0) {};
   \node[ver] at (6,0) {};
   \node[ver] at (7,0) {};
   \node[ver] at (8,0) {};
   \node      at (1,-0.6) {2};
   \node      at (2,-0.6) {5};
   \node      at (3,-0.6) {8};
   \node      at (4,-0.6) {1};
   \node      at (5,-0.6) {3};
   \node      at (6,-0.6) {4};
   \node      at (7,-0.6) {6};
   \node      at (8,-0.6) {7};
   \draw[edg] (1,0) to[bend left=60] (5,0);
   \draw[edg] (5,0) to[bend left=60] (6,0);
   \draw[edg] (6,0) to[bend left=60] (8,0);
   \draw[edg] (2,0) to[bend left=60] (3,0);
 \end{tikzpicture}.
\end{align}

The cover relation is easily described in this representation.  To obtain $(\pi',\sigma')$ such that $(\pi',\sigma')\lessdot(\pi,\sigma)$, choose $\pi'\in NC_n$ such that $\pi'\lessdot \pi$, and $\sigma'$ is obtained by rearranging the labels so as to respect the increasing condition on the blocks of $\pi'$.

We refer to \cite{macdonald} for characters of the symmetric group, symmetric functions, and the Frobenius characteristic map $\chi$ relating these two notions.  The character of $\mathfrak{S}_n / \mathfrak{S}_n(\pi)$ is $\operatorname{Ind}_{\mathfrak{S}_n(\pi)}^{\mathfrak{S}_n}(1)$, the trivial character of $\mathfrak{S}_n(\pi)$ induced to $\mathfrak{S}_n$.  It is such that
\[
  \chi\big( \operatorname{Ind}_{\mathfrak{S}_n(\pi)}^{\mathfrak{S}_n}(1) \big) 
  =
  h_\lambda
\]
where $h_\lambda$ is the {\it homogeneous symmetric function}, and $\lambda$ is the integer partition obtained by sorting block sizes of $\pi$.  (One can replace $\operatorname{Ind}_{\mathfrak{S}_n(\pi)}^{\mathfrak{S}_n}(1)$ with $h_\pi = h_\lambda$ to work with symmetric functions rather than characters in what follows, at the condition of being aware that the evaluation of a character $\Psi$ at some permutation $\sigma$ is $\langle \chi(\Psi) | \frac{p_\lambda}{z_\lambda} \rangle$ where the integer partition $\lambda$ is the cycle type of $\sigma$.)

\begin{prop} \label{prop:2pp_park}
  The character of the action of $\mathfrak{S}_n$ on the orbit of $(\pi , \rho , \lambda) \in \pp_n$ is $\operatorname{Ind}_{\mathfrak{S}_n(\pi)}^{\mathfrak{S}_n}(1)$.  Moreover, character of the action of $\mathfrak{S}_n$ on $\pp_n$ is
  \begin{align} \label{eq:parking}
    \sum_{\pi \in NC_n} \operatorname{Ind}_{\mathfrak{S}_n(\pi)}^{\mathfrak{S}_n}(1).
  \end{align}
\end{prop}

\begin{proof}
From Proposition~\ref{prop:equivariantbij}, we see that the orbit of $(\pi , \rho , \lambda) \in \pp_n$ is isomorphic (as a $\mathfrak{S}_n$-set) to $\mathfrak{S}_n / \mathfrak{S}_n(\pi)$.  The first statement follows.  The second one is obtained by summing over the orbits in $\pp_n$.
\end{proof}

The character in \eqref{eq:parking} coincides with that of the {\it noncrossing parking space} from \cite{ARR}.  The evaluation of this character is given by $\sigma \mapsto (n+1)^{z(\sigma)-1}$ where $z(\sigma)$ is the number of cycles of $\sigma$.  Let us present a more general version of this character, defined by Rhoades~\cite{rhoades}.  
He considered a character which comes from the action of $\mathfrak{S}_n$ on chains of $\pp_n$ (in particular, the poset $\pp_n$ appears implicitly in \cite{rhoades}).  Note that via Lemma~\ref{lem:unique}, the orbit of a chain $\phi_1 \leq \dots \leq \phi_k$ is isomorphic as a $\mathfrak{S}_n$-set to the orbit of $\phi_k$.

\begin{prop}[{\cite[Section~8]{rhoades}}]
Let $\operatorname{Park}^{(k)}_n$ denote the character of $\mathfrak{S}_n$ acting on $k$-chains $\phi_1 \leq \dots \leq \phi_k$ of $\pp_n$, so that:
\[
  \operatorname{Park}^{(k)}_n
  = 
  \sum_{\substack{\pi_1 \leq \dots \leq \pi_k \\ \pi_1 , \dots , \pi_k \in NC_n}}
  \operatorname{Ind}_{\mathfrak{S}_n(\pi_k)}^{\mathfrak{S}_n}(1).
\]
Then, for any $\sigma\in \mathfrak{S}_n$, we have:
\begin{align}  \label{char_chains}
  \operatorname{Park}^{(k)}_n(\sigma) 
  = 
  (kn+1)^{z(\sigma)-1}.
\end{align}
\end{prop}

This character $\operatorname{Park}^{(k)}_n$ will be called a {\it zeta character}, as it encompasses both the zeta polynomial of $\pp_n$ (see Section~\ref{sec_enum}) and its character as a $\mathfrak{S}_n$-set.  There is an alternative expression for $\operatorname{Park}^{(k)}_n$, that will be useful below.

\begin{lemm} \label{lemm_parkk}
We have:
\begin{equation}  \label{eq_parkk}
  \operatorname{Park}^{(k)}_n 
  = 
  \sum_{\pi\in NC_n} \bigg( \prod_{b\in K(\pi)} C^{(k)}_{|b|} \bigg) 
  \operatorname{Ind}_{\mathfrak{S}_n(\pi)}^{\mathfrak{S}_n}(1).
\end{equation}
\end{lemm}

\begin{proof}
Let $(\pi_1,\rho_1,\sigma_1) \leq \dots \leq (\pi_k,\rho_k,\sigma_k)$ be a $k$-element chain in $\pp_n$.  As noted above, the action of $\mathfrak{S}_n$ on its orbit is isomorphic to the action on the orbit of $(\pi_k,\rho_k,\sigma_k)$.  This follows from Lemma~\ref{lem:unique} and the fact that the group only acts on the second and third elements of each triple.  Therefore, the character of this action is $\operatorname{Ind}_{\mathfrak{S}_n(\pi_k)}^{\mathfrak{S}_n}(1)$.

For a given $\pi\in NC_n$, the number of $k$-element chains having $\pi$ as their top element is equal to 
\[
  \prod_{b\in K(\pi)} C^{(k)}_{|b|},
\]
indeed this follows from the result on the structure of order ideals in $NC_n$ and knowing its zeta polynomial.  Thus, by summing over $\pi$ we get the desired formula for $\operatorname{Park}^{(k)}_n$.
\end{proof}

\subsection{Classical parking functions}
\label{sec13}

In this section, we use the following terminology.  A {\it weak composition of an integer} $n\geq 0$ is a finite sequence of nonnegative integers, such that the sum is $n$.  Similarly, a {\it weak composition of a set} $X$ is a finite sequence of pairwise disjoint sets (that are possibly empty), such that the union is $X$.

We first need a lemma that gives an alternative encoding of noncrossing partitions.  

\begin{lemm}  \label{lemma_encodingNC}
  There is a bijection between $NC_n$ and weak compositions $(a_1,\dots,a_n)$ of $n$ such that $\sum_{i=1}^j a_i \geq j $ for any $j \in \llbracket 1;\dots,n\rrbracket$.  It is given by:
  \[
     a_i = \begin{cases}
               |B| & \text{ if } i = \min (B) \text{ for some } B\in \pi, \\
               0 & \text{ otherwise}.
           \end{cases}
  \]
\end{lemm}

\begin{proof}
  We only give a quick description of the inverse bijection, and details are left as an exercise.  First, there is a bijection between compositions as in the proposition and Łukasiewicz paths of length $n$ (which are, by definition, lattice paths in $\mathbb{N}^2$ from $(0,0)$ to $(n,0)$ with steps of the form $(1,i)$ for $i\in\mathbb{N}\cup\{-1\}$).  Explicitely, to $(a_1,\dots,a_n)$ we associate the lattice path with steps $(1,a_1-1), \dots, (1,a_n-1)$.  We build a noncrossing partition from a Łukasiewicz path as follows:
  \begin{itemize}
      \item if the $i$th step is $(1,0)$ then there is a block $\{i\}$,
      \item if the $i$th step is $(1,j)$ with $j>0$, and the $j$ facing steps $(1,-1)$ have indices $i_1,\dots,i_j$, then there is a block $\{i,i_1,\dots,i_j\}$.
  \end{itemize}
  The notion of facing steps in a path is illustrated by the horizontal arrows in Figure~\ref{fig_luk}.  The noncrossing partition associated to the path in this figure is $\{\{1,2,15\}, \allowbreak \{3,6,10,11\}, \allowbreak \{4,5\}, \allowbreak \{7,8,9\}, \allowbreak \{12,13,14\}\}$.
\end{proof}
  
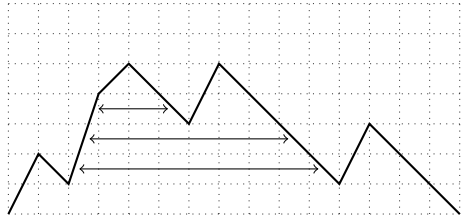
\begin{figure}[h!tp]
  \centering
  \begin{tikzpicture}[scale=0.4]
    \draw[dotted,darkgray] (0,0) grid (15,7);
    \draw[thick] (0,0) -- (1,2) -- (2,1) -- (3,4) -- (4,5) -- (5,4) -- (6,3) -- (7,5) -- (8,4) -- (9,3) -- (10,2) -- (11,1) -- (12,3) -- (13,2) -- (14,1) -- (15,0);
    \draw[<->] (3,3.5) -- (5.3,3.5);
    \draw[<->] (2.7,2.5) -- (9.3,2.5);
    \draw[<->] (2.36,1.5) -- (10.3,1.5);
  \end{tikzpicture}
  \caption{A Łukasiewicz path.\label{fig_luk}}
\end{figure}

The following is the classical definition of a parking function.

\begin{defi} \label{defi_classicalpp}
A {\it parking function} of length $n$ is a word $w_1\dots w_n$ of positive integers, such that for all $k$ between $1$ and $n$, we have $ \# \{ \; i  \; : \; w_i\leq k  \} \geq k$ (equivalently, the increasing sort of $w_1\dots w_n$ is below $1,2,\dots,n$, entrywise). The symmetric group acts on parking functions in a natural way: for $\sigma\in\mathfrak{S}_n$,
$ \sigma \cdot ( w_1\dots w_n ) = w_{\sigma^{-1}(1)} \dots w_{\sigma^{-1}(n)}$.
\end{defi}

A parking function of length $n$ can be rewritten as a weak set composition $(A_1, \ldots, A_n)$ of $\llbracket 1;n\rrbracket$ satisfying $\sum_{i=1}^k |A_i| \geq k$ for any $1 \leq k \leq n$.  The correspondence is done by letting $A_i$ be the set of positions of $i$ in the parking function: $A_i = \{ j \; | \; w_j =i \}$.  In this reformulation, the action of $\mathfrak{S}_n$ on parking functions naturally extends the action of $\mathfrak{S}_n$ on subsets of $\llbracket 1;n\rrbracket$.

By seeing 2-noncrossing partitions as ``enriched'' noncrossing partitions, the previous lemma can be extended to give the following:

\begin{prop} \label{bij_2ncp_park}
  There is a $\mathfrak{S}_n$-equivariant bijection between $\pp_n$ and parking functions of length $n$, defined by the following property: the image of $(\pi,\rho,\lambda)\in\pp_n$ is $w_1 \dots w_n$ such that
  \[
    \forall B\in\pi, \; \forall i\in \lambda(B), \;  w_i=\min B.
  \]
  Equivalently (in terms of weak set compositions), this bijection sends $(\pi,\rho,\lambda)$ to $(A_1,\dots,A_n)$ such that:
  \[
     A_i = \begin{cases}
               \lambda(B) & \text{ if } i = \min (B) \text{ for some } B\in \pi, \\
               \varnothing & \text{ otherwise}.
           \end{cases}
  \]
\end{prop}

\begin{proof}
It is straightforward to check that the two formulations are equivalent, so we only prove the second one.

It is known that both sets have the same cardinality.  To describe the inverse bijection, first note that the sequence $(|A_1|,\dots,|A_n|)$ is the weak composition of $n$ corresponding to $\pi$ via the bijection in Lemma~\ref{lemma_encodingNC}, since $|\lambda(B)| = |B|$ for $B\in \pi$.  Second, note that $\rho = \{A_i \; | \; A_i\neq \varnothing\}$.  Eventually, $\lambda$ is such that $\lambda(B) = A_{\min(B)}$.  The map in the proposition is thus injective, and bijective.
\end{proof}

For example, this maps sends the 2-noncrossing partition in~\eqref{pfp_example} to the parking function 41112712, with 1s in position 2347, etc.

It is worth making explicit what are the parking functions corresponding to $(\pi,\pi,id)$, because these are orbit representatives.  We will not use this result hereafter, and the proof is left as an exercise.

\begin{prop}
  The bijection from the previous proposition sends the elements $(\pi,\pi,id) \in \pp_n$ to parking functions $w_1\dots w_n$ such that:
  \begin{itemize}
      \item $w_i\leq i$ for all $i\in\llbracket 1;n\rrbracket$,
      \item $w_1\dots w_n$ is lexicographically maximal among parking functions in the same orbit and satisfying the previous condition.
  \end{itemize}
\end{prop}

\begin{rema}
It seems there is no clear and simple way to describe the poset structure of $\pp_n$ directly in terms of words $w_1,\dots,w_n$ as in Definition~\ref{defi_classicalpp}. Indeed, the covering relation there are given by choosing a letter which appears more than twice and increasing it. The main difficulty lays in the value to which it is allowed to increase the letter. \end{rema}

\subsection{Species and parking trees}

We refer to \cite{BLL} for the notion of {\it combinatorial species} and operations on them.  

\begin{defi}
A \emph{parking tree} on a set $L$ is a rooted plane tree $T$ such that:
\begin{itemize}
 \item internal vertices of $T$ are labelled with nonempty subsets of $L$, which form a set partition of $L$,
 \item leaves are labelled by empty sets,
 \item each vertex has as many children as elements in its label.
\end{itemize}
The \emph{species of parking functions} (or \emph{parking species}), denoted $\mathcal{P}_f$, is the species which associates to any finite set $L$ the set of parking trees on $L$ as above.
\end{defi}

Note that a parking tree on $L$ has $\#L$ edges.

\begin{exam} We represent below the parking trees on $\{1\}$, $\{1, 2\}$ and $\{1, 2, 3\}$. We left blanks for leaves:
\begin{align*}
\includegraphics{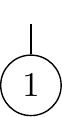}, \quad
\includegraphics{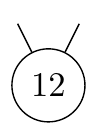}, \;
\includegraphics{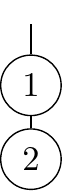}, \;
\includegraphics{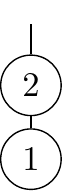}, \quad
\includegraphics{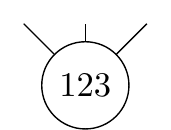}, \;
\includegraphics{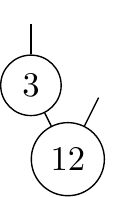}, \;
\includegraphics{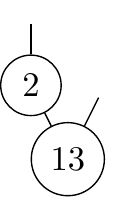}, \;
\includegraphics{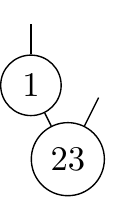}, \;
\includegraphics{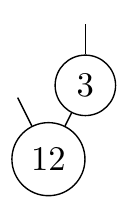},\\
\includegraphics{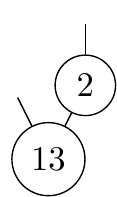}, \;
\includegraphics{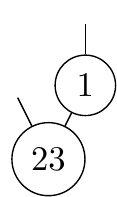}, \;
\includegraphics{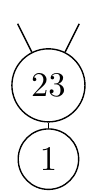}, \;
\includegraphics{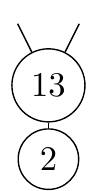}, \;
\includegraphics{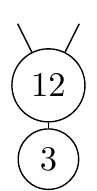}, \;
\includegraphics{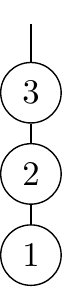}, \;
\includegraphics{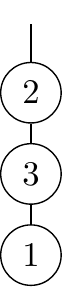}, \;
\includegraphics{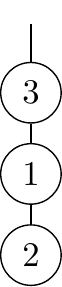}, \;
\includegraphics{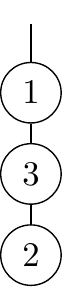}, \,
\includegraphics{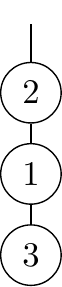}, \;
\includegraphics{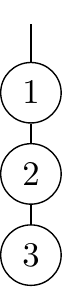}.
\end{align*}
\end{exam}

\begin{figure}
\centering
1325271 
\hspace{2mm}
$\leftrightarrow$
\hspace{2mm}
$\scriptstyle{(\{1,7\}, \{3,5\}, \{2\}, \emptyset, \{4\}, \emptyset, \{6\})}$
\hspace{2mm}
$\leftrightarrow$
\hspace{2mm}
\begin{tikzpicture}[grow=up, scale=0.7, sibling distance=24mm, level distance=12mm]
\node[draw, ellipse]{1\ 7}
    child {     
        node[draw, circle] {\small 6} child{}}
    child{
            node[draw, ellipse] {\small{3 \ 5}}
        		child{node[draw, circle] {\small 4} child{}
        		} 
 	       		child{node[draw, circle] {\small 2}
        		child{} 
        		} 
        }              
;
\end{tikzpicture}
\caption{Bijection between parking functions, weak set compositions and parking trees.\label{FigBijParkFunctTree}}
\end{figure}
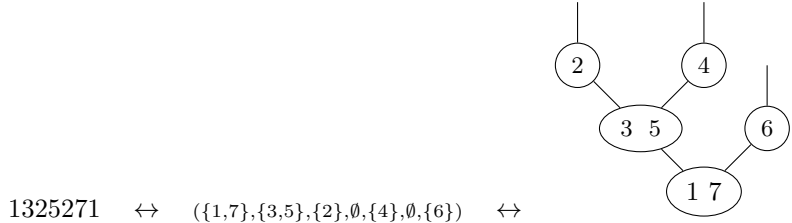

\begin{lemm}\label{BijParkNCP} There exists an explicit bijection between 
the set of parking trees $\mathcal{PT}_n$ on $\llbracket 1;n \rrbracket$ 
and the set $\mathcal{P}_n$ (as in Section~\ref{sec_parkingspace}) which preserves the action of the symmetric group.  
\end{lemm}

\begin{proof}
  We construct a family of bijections 
  $\phi_n : \mathcal{P}_n \rightarrow \mathcal{PT}_n$ by induction on 
  $n\geq 0$. When the index of $\phi_k$ will be obvious, we will just 
  omit it.
  
  If $n=0$, $\mathcal{P}_n$ contains only one pair 
  $(\{\}, \operatorname{id} \in \mathfrak{S}_0)$ and $\mathcal{PT}_n$ contains only 
  one tree, which is the empty parking tree. The map $\phi_0$ is 
  defined by associating the elements in both sets.
  
  Suppose that we have constructed bijections $\phi_k$ for $k<n$. We 
  now construct $\phi_n$. Let us consider a pair $p=(\pi, \sigma)$ in 
  $\mathcal{P}_n$. The noncrossing partition $\pi$ admits the 
  following decomposition, each "arch" of the non-crossing partition 
  deliminating a non-crossing partition: 
  \begin{itemize}
      \item a part $E$ (the first of $\pi$) of size $\ell$
      \item $\ell$ (possibly empty) non-crossing partitions $\pi_1, \ldots, \pi_{\ell}$ of size strictly less than $n$
  \end{itemize}
  By induction hypothesis, we can associate to each $(\pi_k, \sigma_{|\pi_k})$ a parking tree $T_k$. We now construct a tree $T$, whose root is given by $\sigma \cdot E$ and has $\ell$ children given from left to right by $T_1, \ldots, T_k$.
  Moreover it is a parking tree as the children of the root are parking trees and the root has a label of cardinality $\ell$ and exactly $\ell$ children.
    \end{proof}
\begin{exam}
The decomposition associated with the noncrossing partition drawn on  \eqref{pfp_example} is a part $\{2,3,4,7\}$ and four noncrossing partitions $\{\{5,8\},\{1\}\}$, $\emptyset$, $\{6\}$ and $\emptyset$. The corresponding parking tree is then:

\begin{center}
\begin{tikzpicture}[grow=up, scale=0.7, sibling distance=24mm, level distance=12mm]
\node[draw, ellipse]{\small 2 \ 3 \ 4\ 7}
    child{}
    child{
            node[draw, circle] {6}
            child{}
        }   
    child{}
    child {     
        node[draw, ellipse] {5 \ 8} 
        child{ node[draw, circle] {1} child{}
        }
        child{}}
        
;
\end{tikzpicture}
 \end{center}
\end{exam}

\begin{lemm} \label{BijParkFunctTree}
There exists an explicit bijection between 
the set of parking trees $\mathcal{PT}_n$ on $\llbracket 1;n \rrbracket $ 
and parking functions
which preserves the action of the symmetric group.  
\end{lemm}

\begin{proof}
First, let us recall that parking functions can be represented as (weak) set compositions satisfying some properties, using the same bijection as in Proposition~\ref{bij_2ncp_park}. 

To a parking function $f=w_1\ldots w_n$ of length $n$ with $k:=\max(w_i, 1 \leq i \leq n)$ can be associated a weak set composition of $\{1, \ldots, n\}$,  $\varphi(f)=(E_1, \ldots, E_k)$ given by $E_i=\{j | w_j=i\}$. The property of $f$ being a parking function can immediately be translated into $\sum_{i=1}^m | E_i| \geq m$, for $1 \leq m \leq k$. Reciprocally, to a weak set composition of $\{1, \ldots, n\}$ $c=(E_1, \ldots, E_k)$ satisfying $\sum_{i=1}^m | E_i| \geq m$, for $1 \leq m \leq k$, can be associated a word of length $n$, $\psi(c)=w_1\ldots w_n$, defined by $w_i=j$ for any $i \in E_j$. The conditions on the cardinality of the sets $E_i$ ensure that the word $\psi(c)$ is indeed a parking function.

To prove the above lemma, we will use the representation of parking functions in terms of weak set compositions.

\begin{description}
\item[From weak set compositions to parking trees : ]
Let us consider a weak set composition of $\llbracket 1 ; n \rrbracket $, $c=(E_1, \ldots, E_{n+1})$ satisfying $\sum_{i=1}^m | E_i| \geq m$, for $1 \leq m \leq n$ (it is always possible to consider weak set compositions of $\llbracket 1 ; n \rrbracket $ of length $n+1$ up to adding some empty sets at the end of the composition). We consider the parking tree $T(c)$ obtained recursively as follows:
\begin{itemize}
    \item the vertices of $T(c)$ are labelled by the sets $E_i$ (hence $T(c)$ has $n+1$ vertices)
    \item the root is labelled by $E_1$ (which is non-empty thanks to the hypothesis)
    \item if $E_i$ is not empty, $E_{i+1}$ is the leftmost child of $E_i$
    \item if $E_i$ is empty, $E_{i+1}$ is grafted as the rightmost child  of the closest ancestor $E_j$ of $E_i$ ($j \leq i$) having strictly less than $|E_j|$ children.  
\end{itemize}
Here, we consider that a node is the \emph{last child of its parent} if it is the $p$th child and the label of its parent as $p$ elements. 

We have to check that:
\begin{enumerate}
    \item it is always possible to graft a vertex as mentioned above,
    \item for each internal vertex, the cardinality of its label coincides with the number of its children.
\end{enumerate}

For the first point, we have to check that such an ancestor $E_j$ exists. By construction, the child of a node $N$ is introduced only when the lineage of its older siblings is full. If every $E_j$ ancestor of $E_i$ is the last child of its parent, it means that every $E_j$, for $j \leq i$ has exactly $|E_j|$ children. Hence the number of vertices of the constructed tree is:
\begin{itemize}
    \item $i$, since vertices are labelled by $E_1, \ldots, E_i$
    \item $\sum_{j=1}^i |E_j|+1$, since every vertex but the root is the child of another vertex. 
\end{itemize}
By hypothesis, $\sum_{j=1}^i |E_j|+1 >i$: we then get a contradiction and the existence of $E_j$ is proven. 

For the second point, the grafting algorithm ensures that an internal vertex does not have more children than the number of elements in its label. Denoting by $c_i$ the number of children of vertex $E_i$, we have $c_i\leq |E_i|$. Moreover, every $E_i$ are grafted so the tree has $n+1$ nodes, i.e. $n+1=\sum_{i=1}^{n+1} c_i +1 \leq \sum_{i=1}^{n+1} | E_i|+1 \leq n+1$. Hence every inequalities are equalities and every $E_i$ has exactly $|E_i|$ children. As it is clear that the obtained graph is a rooted planar tree, $T(c)$ is then a parking tree.

\item[From parking trees to weak set compositions : ]

Reading nodes of a parking tree $t$ in the prefix order gives a weak set composition $C(t)=(E_1, \ldots, E_{n+1})$. The node labelled by $E_m$ is the $m$th visited node, for $2 \leq m \leq n+1$. When we visit $E_m$ for the first time, all visited nodes ($E_m$ included) are $E_1$ or descendants of some $E_i$, for $1 \leq i < m$. This leads to the inequality $\sum_{i=1}^{m-1} | E_i| +1 \geq m$, hence the set composition satisfies the desired hypothesis.
\end{description}
Finally, the map $C$ is the inverse bijection of $T$. Indeed, the grafting order of the map $T$ corresponds to the prefix order. See Figure \ref{FigBijParkFunctTree} for an example of these bijections.
\end{proof}

\begin{rema}
Suppose that $L$ (and consequently, any of its subsets) is endowed with a total order.  Since parking trees are planar, we can for each vertex canonically associate each outgoing edge to an element in the vertex. In this case we can associate each internal vertex that is not the root to an element of $L$ from which it descends. This defines a function $f : L \mapsto L$, that is not defined on the elements of the root, constant on each node of the tree, and which is nilpotent. Conversely, for each nilpotent partial function $f$ on $L$, we can recursively build a parking tree: the root is the set of elements on which $f$ is not defined, and for each $i \in L$, its descendant is labelled by $f^{-1}(i)$. Note that if $|\operatorname{Im}(f)| = r$, then the resulting tree has $r+1$ vertices. This bijection falls in line with Laradji and Umar's in \cite{laradjiumar}, in which they enumerate partial nilpotent functions with a fixed image set of size $r$.
\end{rema}

The covering relations on parking trees corresponding to those of the $2$-partition poset are then given as follows.  From a parking tree $T$, another one $U$ such that $T\lessdot U$ is obtained from $T$ by a sequence of operations, represented on Figure~\ref{fig:exple_cover}:
\begin{itemize}
\item choose a vertex $A$ and  partition it into two (non empty) sets $A_1$ and $A_2$,
\item deconcatenate the list of its (possibly empty) subtrees into three lists $L_1$, $L_2$ and $L_3$, such that $L_1$ is non empty and $L_2$ and $A_2$ have the same cardinality,
\item remove from the tree the elements of $A_2$ and $L_2$
\item add the elements of $A_2$ to the rightmost leaf of $A_1$ in $L_1$
\item add $L_2$ as the list of children of $A_2$.
\end{itemize} 

For the leftmost tree in Figure \ref{fig:exple_cover}, $A_1=\{1,5,6\}$ and $A_2=\{2\}$, the possible lists are $(L_1, L_2, L_3) = ((\textcolor{blue}{\emptyset}),(\includegraphics{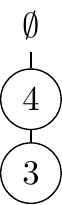}),(\textcolor{red}{\emptyset}, \textcolor{green}{\emptyset}))$, $((\textcolor{blue}{\emptyset},\includegraphics{34.pdf}),(\textcolor{red}{\emptyset}), (\textcolor{green}{\emptyset}))$ or $((\textcolor{blue}{\emptyset},\includegraphics{34.pdf},\textcolor{red}{\emptyset}), (\textcolor{green}{\emptyset}),())$, which gives each of the other trees in Figure \ref{fig:exple_cover}.

\begin{figure}[h!tp]
    \centering
    \includegraphics[scale=0.7]{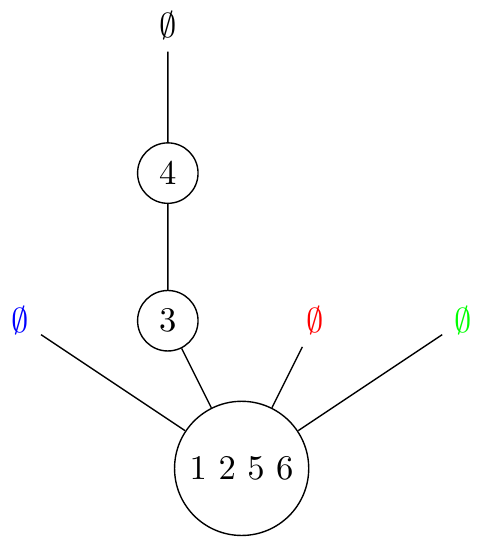}
    \qquad
    \includegraphics[scale=0.7]{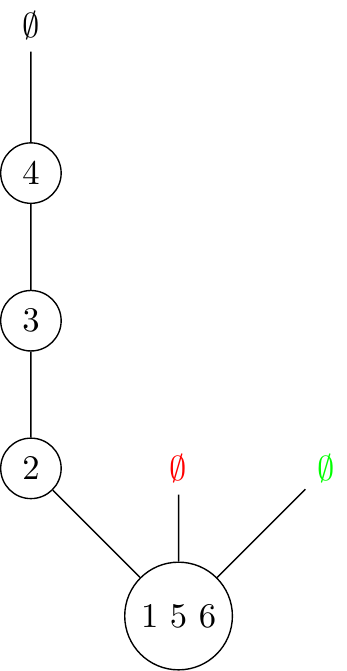}
    \qquad
    \includegraphics[scale=0.7]{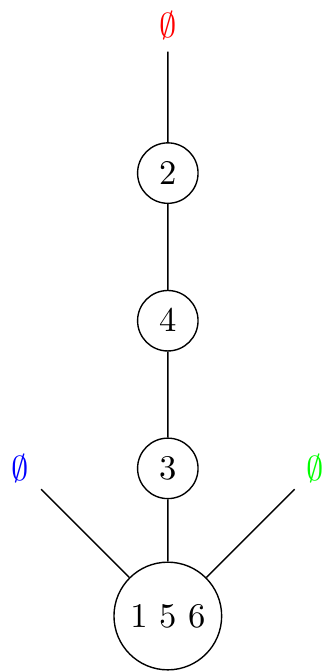}
    \qquad
    \includegraphics[scale=0.7]{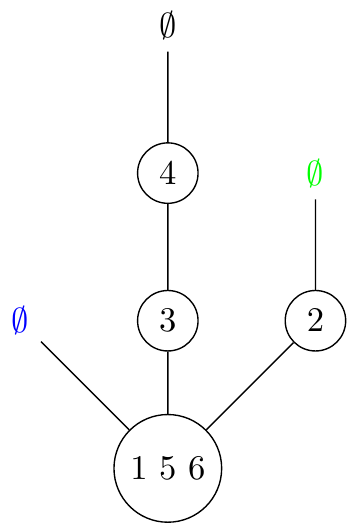}
    \caption{A parking tree and some parking trees covering it.\label{fig:exple_cover}}

\bigskip

\includegraphics[scale=0.8]{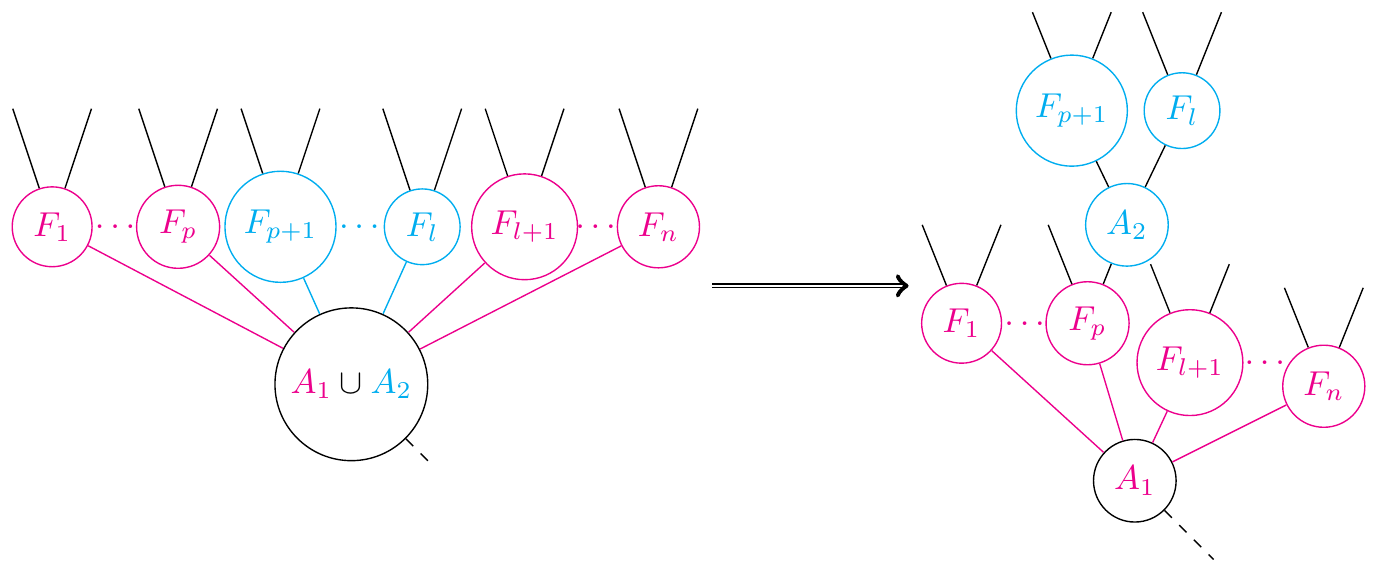}
\caption{Covering relations in parking trees poset.\label{fig_covers}}
\bigskip

\includegraphics[scale=0.8]{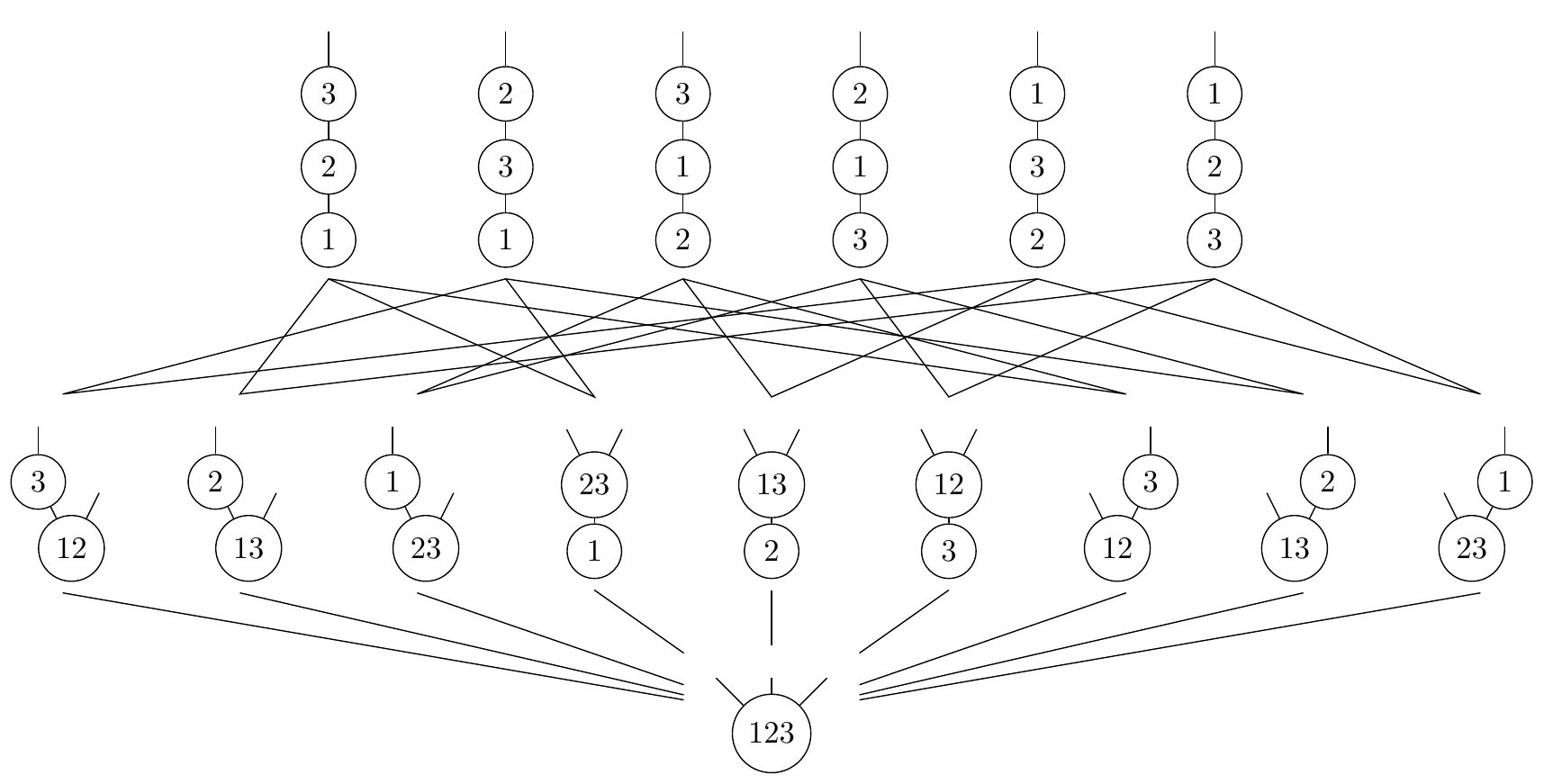}
\caption{The poset of parking trees on three elements.\label{fig_parkpos3}}
\end{figure}

\subsection{Face poset of the permutahedron} \label{permutohedron}

We study in this subsection the restriction of parking poset to elements of the poset such that the associated non-crossing partitions are interval partitions. In terms of parking trees, it corresponds to parking trees which are right combs. We detail below this bijection:

\begin{lemm} \label{Lem_permutahedron}
  The following objects are in bijection:
  \begin{itemize}
      \item 2-partitions whose associated non-crossing partition is an interval partition in $k$ parts,
      \item parking trees which are right combs with $k$ inner vertices,
      \item composition of $k$ non-empty sets.
  \end{itemize}
\end{lemm}

\begin{proof}
We first consider 2-partitions whose associated non-crossing partitions are interval partition. Let us consider such a 2-partition $(\pi, \rho, \lambda)$. If $B_1, \ldots, B_k$ are the blocks of $\pi$, ordered according to their minimal elements and $b_i$ the size of $B_i$,we then have $\max B_i < \min B_{i+1}$, for any $i$ between $1$ and $k-1$. Moreover, the associated parking function has its value in $\{1, b_1+1, \ldots, \sum_{i=1}^{j} b_i +1, \ldots, \sum_{i=1}^{k-1} b_i +1\}$. Let us show by induction that the associated parking tree is a right comb. If $k=1$, the 2-partition is trivial and has only one part. The associated parking tree has only one node containing all the labels $\{1, \ldots, n\}$ : it is a right comb. Let us now suppose the property true up to $k-1$ and consider a 2-partition with $k$ parts $(\pi, \rho, \lambda)$, $\pi$ being an interval non-crossing partition. Considering the non-crossing partition $\tilde{\pi}=\{B_1, \ldots, B_{k-1}\}$, we can associate to it a 2-partition $(\tilde{\pi}, \{\lambda(B_1), \ldots, \lambda(B_{k-1})\}, \lambda)$ on which we can apply the induction hypothesis: the associated parking tree is a right comb. To build the parking tree associated with $\pi$, we have to add the vertex associated with the value of $B_k$ and containing labels of $\lambda(B_k)$. As the non-crossing partition is an interval partition, the value associated with $B_k$ is $\sum_{i=1}^{k-1} b_i$. The only way for the node associated with $B_k$ to be associated with this value is to graft it on the rightmost leaf of the tree: we then get a right comb. Let us denote by $\phi$ the map sending 2-partitions whose associated non-crossing partition is an interval partition to parking trees which are right combs.

When reading nodes on the right branch of a parking tree which are a right comb, we get a composition of sets which are known to be equivalent to surjection. We denote this map by $\psi$. This map is clearly a bijection. If we denote by $C=(C_1, \ldots, C_k)$ the composition, the $i$th set of the composition is the $C_i$th child in the parking tree of the $i-1$th set of the composition, the other children being empty.

To finish the proof, we now have to describe the map $\mu$ which associates to a composition of sets $C=(C_1, \ldots, C_k)$ a 2-partition whose associated non-crossing partition is an interval partition. The set partition is obtained by forgetting the order of the composition. The non-crossing partition is $\{\{1, \ldots, |C_1|\}, \{|C_1|+1, \ldots, |C_1|+|C_2|\}, \ldots, \{\sum_{p=1}^{k-1}|C_p|+1, \ldots, \sum_{p=1}^{k}|C_p|\}\}$. The map $\lambda$ is given by associated to the part which is the $i$th of $C$ the $i$th part of the non-crossing partition. 

The map $\phi \circ \mu \circ \psi$ sends a parking tree which is a right comb to a parking tree which is a right comb. Moreover, the $k$th node of the tree is send by $\psi$ to the $k$th set of the composition which is send by $\mu$ to the $k$th part of the noncrossing partition, itself send by $\phi$ to the $k$th node of the parking tree. Hence, $\phi \circ \mu \circ \psi$ is the identity and $\phi$ and $\mu$ are also bijections. 
\end{proof}

Compositions of sets are known to label the faces of the permutahedron. The link between parking trees and the permutahedron goes deeper as the order on parking tree studied in this article is exactly the inclusion order of the permutahedron.

\begin{prop}
  The subposet of $\pp_n$ of parking trees which are right combs is isomorphic to the face poset of the permutahedron. 
\end{prop}
  
  \begin{proof}
The allowed covering relations are when a vertex $A=A_1 \sqcup  A_2$ is split into $A_1$ and $A_2$, with $A_2$ grafted on the rightmost leaf of the $|A_1|$th child of $A_1$ and have children of $A$ of indices between $|A_1|+1$ and $|A_1|+|A_2|$ as  children. As all children but the last one are empty, the covering relations can be simplified in choosing a node $A$ with rightmost child $C$ and replacing the right comb  $A-C$ by a right comb $A_1-A_2-C$, where $P-C$ is a shortcut for "$C$ is the rightmost child of $P$ in the tree". This covering relation translating in term of composition of set is exactly the one of the face poset of the permutahedron. In this order, the composition $(C_1, \ldots, C_{l-1}, C_l\cup C_{l+1},C_{l+2} \ldots, C_k)$ is covered by the composition $(C_1, \ldots, C_{l-1} C_l,C_{l+1}, C_{l+2}, \ldots, C_k)$
  \end{proof}

\section{Shellability of the parking functions poset}
\label{sec_shell}

Recall that $\hat\pp_n$ is the bounded poset obtained by adding a new maximal element $\hat 1$ on top of $\pp_n$.  The goal of this section is to build a {\it shelling} of this poset.  This shellability property of $\pp_n$ has geometric consequences, in particular it will be used to study its homology in Section~\ref{sec_homology}.  We refer to~\cite{bjorner80,wachs} for general notions of combinatorial topology (shellability, EL-labelings, etc.)

\subsection{Construction of the shelling}

Let $P$ be a ranked poset of length $n$ (this means that all maximal chains of $P$ contains $n+1$ elements).  A maximal chain of $P$ will be denoted $\underline{p}=(p_i)_{0\leq i \leq n}$, where it is understood that $p_0$ is the minimal element, $p_n$ is the maximal element, and $\forall i\in \llbracket 0;n-1\rrbracket , \; p_i\lessdot p_{i+1}$.

\begin{defi}  \label{def_shelling}
  A {\it shelling} of $P$ is a total order $<$ on its maximal chains, such that if $\underline{p}'<\underline{p}$, there exists a maximal chain $\underline{p}''$ such that (seeing chains of the poset as subsets):
  \begin{itemize}
      \item $\underline{p}''<\underline{p}$
      \item $\# ( \underline{p}'' \cap \underline{p}  ) = n-1$,
      \item $\underline{p}' \cap \underline{p} \subset \underline{p}''$.
  \end{itemize}
\end{defi}

Let us mention that more generally, shellings can be defined for simplicial complexes.  The definition above correspond to shellings of the order complex $\Omega(P)$ (see Section~\ref{sec_homology}).

For $x\in P$, its set of upper covers is:
\[
  \operatorname{Up}(x) := \{ y\in P \;:\; y \gtrdot x\}.
\]
Consider a family $(\prec_x)_{x\in P}$ where each $\prec_x$ is a total order on $\operatorname{Up}(x)$.  This data gives rise to a total order $<_{\rm lex}$ on maximal chains of $P$, using lexicographic comparison: a chain $\underline{p'} = (p'_i)_{0\leq i\leq n}$ precedes another chain $\underline{p} = (p_i)_{0\leq i\leq n}$ if $p'_{j+1} \prec_{p_j} p_{j+1}$ where $j$ is the minimal index such that $p'_{j+1} \neq p_{j+1}$.  This kind of structure is natural in the context of lexicographic shellability.  For example, via the notion of {\it recursive atom ordering}~\cite{wachs} there is a criterion on $(\prec_x)_{x\in P}$ that ensures that $<_{\rm lex}$ is a  shelling of $P$.  However we will use another method.

Suppose that the poset $P$ is endowed with an edge-labeling, i.e., a function $\lambda(x,y)$ defined for each cover relation $x\lessdot y$ and taking values in a totally ordered set.  This gives rise to total orders $\prec_x$ as above via the rule $y \prec_x y'$ iff $\lambda(x,y) \leq \lambda(x,y')$ (we do not discuss here the possibility that $y\neq y'$ and $\lambda(x,y)= \lambda(x,y')$).  An interesting class of labelings whose associated lexicographic order are shellings are EL-labelings.  They are defined as labelings such that each interval contains a unique strictly increasing chain, and it is the lexicographically minimal one (see~\cite{wachs} for details).  Björner and Edelman showed that $NC_n$ admit such an EL-labeling (see~\cite{bjorner80}).  Moreover, there exists an EL-labeling $\lambda$ having the additional property that $\forall y,y' \in \operatorname{Up}(x)$, we have $y\neq y' \Rightarrow \lambda(x,y) \neq \lambda(x,'y)$.  Explicitly, we can take $\lambda(x,y) = \bar x ^{-1} \bar y$ (note that $x\lessdot y$ in $NC_n$ implies that $\bar x ^{-1} \bar y \in \mathfrak{S}_n$ is a transposition, and we order them with the lexicographic order on pairs $(i,j)$ such that $i<j$).  We will use such an EL-labeling $\lambda$ of $NC_n$ as an ingredient to show that $\pp_n$ is shellable.

\begin{defi}
The \emph{code} of a permutation $\sigma = \sigma_1 \dots \sigma_n \in\mathfrak{S}_n$ is the word $\gamma(\sigma) = c_n(\sigma) \allowbreak \ldots \allowbreak  c_1(\sigma)$, with $c_{i}(\sigma) = \# \{ \; j \in\llbracket 1;i-1\rrbracket \; | \; \sigma^{-1}(j)>\sigma^{-1}(i) \; \}$. 
\end{defi}

Concretely, $c_i(\sigma)$ is the number of integers smaller than $i$ on its right in $\sigma$. For instance, $\gamma(15324)= 30100$.  In the sequel, it will be convenient to see elements of $\pp_n$ as pairs in $NC_n\times \mathfrak{S}_n$, as explained in Section~\ref{sec_parkingspace}.  Whenever an element of $\pp_n$ is written as a couple rather than a triple, it is understood that we take this convention.  Moreover, if $\phi=(\pi,\sigma)\in\pp_n$, by convention we write $\gamma(\phi) = \gamma(\sigma)$, $c_i(\phi)=c_i(\sigma)$, etc.

\begin{defi} \label{def:orders}
For each $\phi = (\pi,\sigma) \in \pp_n$ with $\rk(\phi)<n-1$, we define a total order $\prec_\phi$ on $\operatorname{Up}(\phi)$ by:
\[
  (\rho,\tau) \prec_\phi (\rho',\tau') 
  \quad \Longleftrightarrow \quad
  \gamma(\tau) < \gamma(\tau'),
  \text{ or }
    \gamma(\tau) = \gamma(\tau')
    \text{ and }
    \lambda( \pi , \rho) < \lambda( \pi , \rho')
\]
where we use the lexicographic order to compare $\gamma(\tau)$ and $\gamma(\tau')$.  If $\rk(\phi)=n-1$, it has a unique cover in $\hat\pp_n$, namely the maximal element $\hat 1$.  In this case, the total order on $\operatorname{Up}(\phi)$ is the obvious one.
\end{defi}

\begin{thm} \label{theo_shelling}
The order $<_{\rm lex}$ defined using the total orders $(\prec_\phi)_{\phi\in \ppp_n}$ from Definition~\ref{def:orders} is a shelling of $\hat\pp_n$.
\end{thm}

The proof relies on the following lemma, which will be proved in the next sections and is illustrated on Figure~\ref{fig:lem_shell}.
 
\begin{lemm} \label{lemma_for_shelling}
  Let $x,y,y',z\in\pp_n$ such that $x\lessdot y \lessdot z$, $x\lessdot y'$, and $y'\prec_x y$.  Then:
  \begin{itemize}
      \item either there exists $y'' \in \pp_n$ such that $x\lessdot y'' \lessdot z$ and $y''\prec_x y$,
      \item or there exists $z'\in\pp_n$ such that $y\lessdot z' \leq y'\vee z$ and $z' \prec_y z$.
  \end{itemize}
\end{lemm}

\begin{figure}
    \centering
\begin{tikzpicture}[very thick]
   \tikzstyle{ver} = [circle, draw, fill, inner sep=0.5mm];
   \node[ver](O) at (0,0){} node[below=2pt of O.south east] {$\hat{0}$};
   \node[ver](x) at (0,2){} node[left=2pt of x.south east] {$x$};
   \node[ver](y) at (1,3){} node[right=2pt of y.south east] {$y$};
   \node[ver](y') at (-1,3){} node[left=2pt of y'.south east] {$y'$};
   \node[ver](y'') at (0,3){} node[above left=2pt of y''.south east] {$y''$};
   \node[ver](z) at (1,4){} node[right=2pt of z.south east] {$z$};
   \node[ver](z') at (0,4){} node[left=2pt of z'.south east] {$z'$};
   \node[ver](1) at (0,6){} node[left=2pt of 1.south east] {$y' \vee z $};
      \node[ver](max) at (0,7){} node[above=2pt of max.south east] {$\hat{1}$};
   \draw[decorate, violet, decoration=zigzag] (O)--(x);
   \draw[decorate, violet, decoration=zigzag] (1)--(max);
   \draw[magenta] (x)--(y');
   \draw[blue] (x)--(y)--(z);
   \draw[orange] (x)--(y'')--(z);
   \draw[decorate, magenta, decoration=zigzag] (y')--(1);
   \draw[decorate, blue, decoration=zigzag] (z)--(1);
   \draw[decorate, green, decoration=zigzag] (z')--(1);
   \draw[green] (z')--(y);
\end{tikzpicture}
    \caption{Hasse diagram illustrating  Lemma~\ref{lemma_for_shelling}, where straight edges are cover relations, wavy edges are general relations and the existence of magenta, blue and violet edges implies the one of either green or orange edges.}
    \label{fig:lem_shell}
\end{figure}
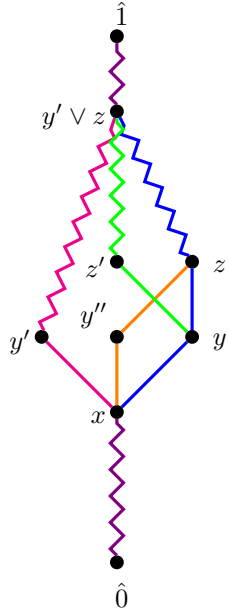
 
\begin{proof}[Proof of Theorem~\ref{theo_shelling}]
  Let $\underline{p} = (p_i)_{0\leq i \leq n} $ and $\underline{p}' = (p'_i)_{0\leq i \leq n} $ be two distinct maximal chains of $\hat\pp_n$ such that $\underline{p}' <_{\rm lex} \underline{p}$.  Our goal is to find $\underline{p}'' = (p''_i)_{0\leq i \leq n} $ as in Definition~\ref{def_shelling}.
  
  Let $x = p_i = p'_i$ where $i$ is the minimal index such that $p_{i+1} \neq p'_{i+1}$. Let $y' = p'_{i+1}$, $y = p_{i+1}$, and note that $\underline{p}' <_{\rm lex} \underline{p}$ means that $y' \prec_x y$. Let $z = p_{i+2}$.  If $i=n-2$, the maximal chains $\underline{p}'$ and $\underline{p}$ only differ in rank $n-1$.  So we can take $\underline{p}'' = \underline{p}'$ and it satisfies the requirements in the definition of a shelling.  Otherwise, we have $i<n-2$, so $z < \hat 1$ and we are in situation to use Lemma~\ref{lemma_for_shelling}.  We distinguish two cases.
  \begin{itemize}
      \item If there exists $y'' \in \pp_n$ such that $x\lessdot y'' \lessdot z$ and $y''\prec_x y$, we define a maximal chain $\underline{p}''$ by replacing $y$ with $y''$ in $\underline{p}$.  By construction, it satisfies the requirements in the definition of a shelling.
      \item Otherwise, there exists $z'\in\pp_n$ such that $y\lessdot z' \leq y' \vee z$ and $z'\prec_y z$.  Let $j$ be the minimal integer such that $p'_j = p_j$ and $j>i$.  We have $j\geq i+3$ (it cannot be equal to $i+2$ because that would mean $y'\lessdot z$, and this situation was already ruled out in the previous case).  Note that $y'$ and $z$ are both below $p_j$, so $y'\vee z \leq p_j $, so $z' \leq p_j$.  Define $\underline{p}'' = (p''_k)_{0\leq k \leq n}$ by: 
      \begin{itemize}
          \item $p''_k = p_k$ if $k\leq i+1$ or $k\geq j$,
          \item $p''_{i+2} = z'$, 
          \item the elements strictly between $p''_{i+2}$ and $p''_j$ are defined arbitrarily with the constraint that $\underline{p}''$ is a maximal chain.
      \end{itemize}
      By construction, we can check: $\underline{p}'' <_{\rm lex} \underline{p}$ (this holds because $z' \prec_y z$), $\underline{p}' \cap \underline{p} \subsetneq \underline{p}'' \cap \underline{p} $ (the inclusion holds because $p_k\neq p'_k $ if $i<k<j$ , it is strict since $y$ is in the second intersection but not in the first one).  By iterating the argument, we can eventually construct a maximal chain satisfying the requirements in the definition of a shelling.
  \end{itemize}
\end{proof} 
 
Now, it remains only to prove Lemma~\ref{lemma_for_shelling}.  

\begin{rema} \label{rem:shellabilitycriterion}
  It would be interesting consider the following problem.  Given a finite poset $P$ and a family $(\prec_x)_{x\in P}$ as above, we have showed that $P$ is shellable if it satisfies a criterion as in Lemma~\ref{lemma_for_shelling} (if $P$ is not a lattice, the condition $z'\leq y'\vee z$ should be replaced with: $\forall u\in P$, $u\leq y'$ and $u\leq z$ implies $u\leq z'$).  It is natural to ask whether this new criterion for shellability can be compared with other ones such as EL-shellability or CL-shellability. From definitions, it is clear that an order on atoms of any intervals in the poset which is a recursive atom ordering satisfies Lemma~\ref{lemma_for_shelling} (more precisely the second case). This criterion is then weaker or equivalent to CL-shellability.
\end{rema}

\begin{rema}
  We wrongly asserted in \cite{fpsac2020} that the orders $(\prec_x)$ form a {\it recursive atom ordering}, a technical condition that characterizes {\it CL-shellability}.  To see that it is not the case, consider $x$ the minimal element together with:
\begin{align*}
 \begin{tikzpicture}[scale=0.6]
   \tikzstyle{ver} = [circle, draw, fill, inner sep=0.5mm]
   \tikzstyle{edg} = [line width=0.6mm]
   \node at (-0.6,0) {$y=$};
   \node[ver] at (1,0) {};
   \node[ver] at (2,0) {};
   \node[ver] at (3,0) {};
   \node[ver] at (4,0) {};
   \node[ver] at (5,0) {};
   \node[ver] at (6,0) {};
   \node      at (1,-0.6) {1};
   \node      at (2,-0.6) {2};
   \node      at (3,-0.6) {4};
   \node      at (4,-0.6) {3};
   \node      at (5,-0.6) {5};
   \node      at (6,-0.6) {6};
   \draw[edg] (1,0) to[bend left=60] (2,0);
   \draw[edg] (2,0) to[bend left=60] (3,0);
   \draw[edg] (4,0) to[bend left=60] (5,0);
   \draw[edg] (5,0) to[bend left=60] (6,0);
   \node at (6.4,-0.1) {,};
 \end{tikzpicture}
 \qquad 
  \begin{tikzpicture}[scale=0.6]
   \tikzstyle{ver} = [circle, draw, fill, inner sep=0.5mm]
   \tikzstyle{edg} = [line width=0.6mm]
   \node at (-0.6,0) {$z=$};
   \node[ver] at (1,0) {};
   \node[ver] at (2,0) {};
   \node[ver] at (3,0) {};
   \node[ver] at (4,0) {};
   \node[ver] at (5,0) {};
   \node[ver] at (6,0) {};
   \node      at (1,-0.6) {1};
   \node      at (2,-0.6) {4};
   \node      at (3,-0.6) {2};
   \node      at (4,-0.6) {3};
   \node      at (5,-0.6) {5};
   \node      at (6,-0.6) {6};
   \draw[edg] (1,0) to[bend left=60] (3,0);
   \draw[edg] (4,0) to[bend left=60] (5,0);
   \draw[edg] (5,0) to[bend left=60] (6,0);
   \node at (6.4,-0.1) {.};
 \end{tikzpicture}
\end{align*}
It is easily checked that $y$ is the $\prec_x$-minimal element among the 2 lower covers of $z$.  However, there exists $z' \in \operatorname{Up}(y)$ such that $z \prec_y z' $, and $z'$ covers an element $y'$ such that $y'\prec_x y$.  Take for example:
\begin{align*}
 \begin{tikzpicture}[scale=0.6]
   \tikzstyle{ver} = [circle, draw, fill, inner sep=0.5mm]
   \tikzstyle{edg} = [line width=0.6mm]
   \node at (-0.6,0) {$y'=$};
   \node[ver] at (1,0) {};
   \node[ver] at (2,0) {};
   \node[ver] at (3,0) {};
   \node[ver] at (4,0) {};
   \node[ver] at (5,0) {};
   \node[ver] at (6,0) {};
   \node      at (1,-0.6) {3};
   \node      at (2,-0.6) {1};
   \node      at (3,-0.6) {2};
   \node      at (4,-0.6) {4};
   \node      at (5,-0.6) {5};
   \node      at (6,-0.6) {6};
   \draw[edg] (1,0) to[bend left=60] (4,0);
   \draw[edg] (2,0) to[bend left=60] (3,0);
   \draw[edg] (4,0) to[bend left=60] (5,0);
   \draw[edg] (5,0) to[bend left=60] (6,0);
   \node at (6.4,-0.1) {,};
 \end{tikzpicture}
 \qquad 
  \begin{tikzpicture}[scale=0.6]
   \tikzstyle{ver} = [circle, draw, fill, inner sep=0.5mm]
   \tikzstyle{edg} = [line width=0.6mm]
   \node at (-0.6,0) {$z'=$};
   \node[ver] at (1,0) {};
   \node[ver] at (2,0) {};
   \node[ver] at (3,0) {};
   \node[ver] at (4,0) {};
   \node[ver] at (5,0) {};
   \node[ver] at (6,0) {};
   \node      at (1,-0.6) {4};
   \node      at (2,-0.6) {1};
   \node      at (3,-0.6) {2};
   \node      at (4,-0.6) {3};
   \node      at (5,-0.6) {5};
   \node      at (6,-0.6) {6};
   \draw[edg] (2,0) to[bend left=60] (3,0);
   \draw[edg] (4,0) to[bend left=60] (5,0);
   \draw[edg] (5,0) to[bend left=60] (6,0);
   \node at (6.4,-0.1) {.};
 \end{tikzpicture}
\end{align*}
\end{rema}

\subsection{A few lemmas}

\begin{lemm} \label{lemma:shelling_nc}
  The statement in Lemma~\ref{lemma_for_shelling} holds when $\pp_n$ is replaced with $NC_n$.
\end{lemm}

\begin{proof}
  Recall that $\lambda$ is an EL-labeling for $NC_n$ (the one in~\cite{bjorner80} for instance).  We only use the existence of $\lambda$, and the fact that for distinct elements $y,y' \in \operatorname{Up}(x)$ we have $\lambda(x,y) \neq \lambda(x,y')$.  Let $x,y,y',z \in NC_n$ as in the statement of Lemma~\ref{lemma_for_shelling}.
  
  If $\lambda(x,y) \geq \lambda(y,z)$, it means that $(x,y,z)$ is a decreasing chain of $[x,z]$. By properties of EL-labelings, it is not the lexicographically minimal chain of $[x,z]$, so there exists $y''$ with $x\lessdot y'' \lessdot z$ and $\lambda(x,y'') < \lambda(x,y)$.  

  Now, assume $\lambda(x,y) < \lambda(y,z)$.  Consider the maximal chain $(p_0,p_1,\dots)$ of $[y , y' \vee z]$ which is strictly increasing.  If $p_1 = z$, by adding $x$ at the beginning we obtain a strictly increasing chain from $x$ to $y' \vee z$ going through $y$ and $z$.  So this chain is lexicographically minimal.  But this is a contradiction: there is a maximal chain from $x$ to $y' \vee z$ going through $y'$, and it is lexicographically smaller than any chain going through $y$.  We thus have $p_1 \neq z$.
  
  Now, let $z' = p_1$.  We have $\lambda(y,z') < \lambda(y,z)$ and $z' \leq y'\vee z$.  This completes the proof.
\end{proof}

\begin{lemm}  \label{lem:invcode}
 Suppose that $\phi,\phi'\in\pp_n$ are such that $\phi\leq\phi'$.  Then $\gamma(\phi) \leq \gamma(\phi')$.
\end{lemm}
 
\begin{proof}
  We only need to consider the case where $\phi'$ covers $\phi$, as the general case follows by transitivity.  Let $\phi=(\pi,\sigma)$ and $\phi'=(\pi',\sigma')$.  Assume $\gamma(\phi) \neq \gamma(\phi')$ (in case of equality, the lemma is true).
 
  Assume that $\gamma(\phi)$ and $\gamma(\phi')$ have a common prefix of length $n-l$, and differ on the $n-l+1$-st letter.  By properties of the code, this means that we have $\sigma^{-1}(i) = \sigma'^{-1}(i)$ for $i\in\llbracket l+1 ;n \rrbracket$ and $\sigma^{-1}(l) \neq \sigma'^{-1}(l)$.  It remains only to show that $\sigma^{-1}(l) > \sigma'^{-1}(l)$, as it easily follows that $c_{l}(\sigma) < c_{l}(\sigma')$, and $\gamma(\sigma) < \gamma(\sigma')$.
  
  As $\pi\lessdot \pi'$, there is a block $b\in \pi$ such that $b\notin\pi'$.  Since $\sigma^{-1}(l) \neq \allowbreak \sigma'^{-1}(l)$, necessarily $\sigma^{-1}(l)$ and $\sigma'^{-1}(l)$ are both in $b$.  
  
  Let $k \in b$ such that $k > \sigma^{-1}(l)$.  As $\sigma$ is increasing along the block $b$ of $\pi$, we have $\sigma(k) > l$, so $k = \sigma'^{-1}( \sigma(k) )$, and $\sigma'(k) = \sigma(k)$.  It follows that $\sigma'(k) \neq l$, and consequently $k\neq \sigma'^{-1}(l)$.  We deduce that $\sigma'^{-1}(l) < \sigma^{-1}(l)$, as needed.
\end{proof}

\begin{lemm}  \label{lemm:veeinpp}
  If $\phi,\phi'\in\pp_n$ are such that $\gamma(\phi)=\gamma(\phi')$, then $\phi\vee\phi' \in \pp_n$ (i.e., $\phi\vee\phi'\neq \hat 1$).  Moreover, $\gamma(\phi\vee\phi') = \gamma(\phi)$.
\end{lemm}

\begin{proof}
  Let $\sigma\in\mathfrak{S}_n$ be the permutation with the same code as $\phi$ and $\phi'$.  We easily see that $\phi$ and $\phi'$ are both below $(1_n,\sigma) \in \pp_n$, and it follows that $\phi\vee\phi' < \hat 1$.  Using the previous lemma, we have:
  \[
    \gamma(\sigma) = \gamma(\phi) \leq \gamma(\phi\vee\phi')\leq \gamma((1_n,\sigma)) = \gamma(\sigma)
  \]
  so that $\gamma(\phi\vee\phi') = \gamma(\phi)$.
\end{proof}

\begin{defi}
  For $\sigma\in\mathfrak{S}_n$, let $p_0(\sigma)\in\llbracket0;n\rrbracket$ be the length of the longest prefix of $\gamma(\sigma)$ containing only zeroes.  If $\phi = (\pi,\sigma) \in \pp_n$, by convention we denote $p_0(\phi)=p_0(\sigma)$.
\end{defi}  

Note that $p_0(\sigma)$ is a decreasing function of $\gamma(\sigma)$.

\begin{lemm}
  For $\phi\in\pp_n$, we have:
  \[
    p_0(\phi)
    =
    \max\big\{ k \in \llbracket 0;n \rrbracket \; : \; \forall i\in \llbracket n-k+1;n\rrbracket, i\in \eta_\phi(i)   \big\}.
  \]
\end{lemm} 
 
\begin{proof}
  We prove that $p_0(\phi) \geq k$ is equivalent to: $\forall i\in \llbracket n-k+1;n\rrbracket, \; i\in \eta_\phi(i)$.
  
  Write $\phi =(\pi,\sigma)$. The condition $p_0(\phi) \geq k$ means $\sigma(i)=i$ for all $i \in \llbracket n-k+1; n \rrbracket$.  For each $i$, $\sigma(i)=i$ implies $i\in \eta_\phi(i)$, since in general we have $\sigma^{-1}(i) \in \eta_\phi(i)$.  So the first condition implies the second one. 
  
  Reciprocally, assume $\forall i \in \llbracket n-k+1;n\rrbracket, \; i \in \eta_\phi(i)$.  By induction on $j \in \llbracket 0; k \rrbracket$, we prove that for all $i \in \llbracket n-j+1; n \rrbracket$, $\sigma(i)=i$.  The case $j=0$ is clear, and the details of the induction are left to the reader.
\end{proof} 

\begin{lemm}
  \label{lemm:samep0}
  If $\phi,\psi\in\pp_n$ are such that $\phi \vee \psi \neq \hat{1}$, then $p_0(\phi \vee \psi) = \min(p_0(\phi),p_0(\psi))$. 
\end{lemm}
 
\begin{proof}
   Since $\phi \vee \psi \neq \hat{1}$, we know from Lemma~\ref{lem:c_eta} that $\eta_{\phi \vee \psi}(k) = \eta_{\phi}(k) \cap \eta_{\psi}(k)$ for $k\in \llbracket 1;n \rrbracket$.  So this is a direct consequence of the previous lemma.
\end{proof}

\begin{defi}  \label{def:Nm}
Let $\phi,\phi'\in\pp_n$ such that $\phi\lessdot\phi'$, with $\phi=(\pi,\sigma)$ and $\phi'=(\pi',\sigma')$.  We define two quantities as follows.
\begin{itemize}
    \item $N(\phi,\phi')$ is the block of $\pi$ that is split to obtain $\pi'$. 
    \item If $\sigma\neq \sigma'$, then $m(\phi,\phi')$ is the maximal $i\in \llbracket 1;n\rrbracket$ such that
    $c_i(\sigma) < c_i(\sigma')$.  Otherwise, $m(\phi,\phi') = 0$.
\end{itemize}
\end{defi}

Remark that if $\phi'_1$ and $\phi'_2$ both cover $\phi$, we have
\[
  m(\phi,\phi'_1) < m(\phi,\phi'_2) \Rightarrow \gamma(\phi'_1) < \gamma(\phi'_2).
\]
Moreover, if $m(\phi,\phi'_1) \neq m(\phi,\phi'_2)$, the reciprocal is true.

\begin{lemm}  \label{lemm:interval_of_4}
  Suppose that $\phi\in \pp_n$ and $\psi,\psi'\in \operatorname{Up}(\phi)$ are such that $N(\phi,\psi)\neq N(\phi,\psi')$.  Then $\psi'\vee\psi$ has rank $\rk(\phi)+2$, and the open interval $(\phi,\psi\vee\psi')$ only contains $\psi$ and $\psi'$.  Moreover:
   \[
     m(\phi,\psi) = m(\psi',\psi\vee\psi'), \quad 
     m(\phi,\psi') = m(\psi,\psi\vee\psi'),
   \]
   and $m(\phi,\psi) \neq m(\phi,\psi')$ unless they are both equal to $0$.
\end{lemm}

\begin{proof}
  This is clear upon inspection.
\end{proof}

\begin{lemm} \label{lemm:case7}
  Let $x,y,y' \in \pp_n$ such that $x\lessdot y$, $x\lessdot y'$, and $N(x,y)=N(x,y')$. Then for any $u,v$ such that $x\leq u \lessdot v \leq y\vee y'$, we have $m(u,v) \leq \max( m(x,y) , m(x,y') )$.
\end{lemm}

\begin{proof}
  We first prove the case where $x=\hat 0$.  We clearly have $m(u,v) \leq n - p_0(v)$.  On the other side, $p_0(v) \geq p_0(y' \vee y) = \min( p_0(y) , p_0(y') )$ by Lemma~\ref{lemm:samep0}.  So $m(u,v) \leq \max( n - p_0(y) , n - p_0(y') )$.  Using $x= \hat 0$, we easily get $m(x,y) = n - p_0(y)$ and similarly for $y'$.  We thus get the desired inequality.
 
  Now, consider the general case ($x\neq \hat 0)$.
  Let $b = N(x,y)$.  All elements in the interval $[x,y\vee y']$ are obtained from $x$ by splitting the block $b$.  We can discard other blocks of $x$ to identify the interval $[x,y'\vee y]$ with an initial interval in $\pp_{n'}$ with $n' = \# b$ (initial means that the bottom element of this interval is the bottom element of the poset).  Via this identification, it is straightforward to see that the quantities $m(x,y)$, $m(x,y')$, etc., are changed via a relabelling which preserves their relative order.  We thus get the result from the case $x = \hat 0$.
\end{proof}

\begin{lemm} \label{lemm:increasingm}
  Suppose that we have $\phi \lessdot \chi \lessdot \psi$ in $\pp_n$.  If there exists no $\chi'$ such that $\phi \lessdot \chi' \lessdot \psi$ and $\chi' \prec_\phi \chi$, then we have $m(\phi,\chi) \leq m(\chi,\psi)$.
\end{lemm}

\begin{proof}
  First, assume that there exists $\chi'$ such that $\psi = \chi \vee \chi'$ and $N(\phi,\chi) \neq N(\phi,\chi') $.  By hypothesis, we have $\chi \prec_\phi \chi'$, so $\gamma(\chi)\leq \gamma(\chi')$ and $m(\phi,\chi) \leq m(\phi,\chi')$. Besides, we have $m(\phi , \chi' ) = m( \chi , \psi )$ by Lemma~\ref{lemm:interval_of_4}.  So $m(\phi , \chi ) \leq m( \chi , \psi )$.
  
  Otherwise, $\psi$ is obtained from $\chi$ by splitting one block.  It means that we can assume $\phi = \hat 0$, using the same argument in the second paragraph of the proof of the previous lemma.  Note that $m(\phi,\chi) = n - p_0(\chi)$ (using $\phi = \hat 0$).  Also, we have $n- p_0(\chi) \leq n-p_0(\psi)$ (since $\chi \leq \psi)$.  
  
  By a way of contradiction, assume $m(\phi,\chi) > m(\chi,\psi)$, so that $m(\chi,\psi) < n- p_0(\psi)$.  It remains to show that there exists $\chi'$ such that $\phi \lessdot \chi' \lessdot \psi$ and $m(\chi',\psi) = n- p_0(\psi)$, indeed it easily follows that $\chi' \prec_\phi \chi$.  To build $\chi'$, denote $n - p_0(\psi) = l$.  It means that in $\psi$, the label $l$ has a smaller label on its right, and labels $l+1,\dots, n $ are on the final part of the permutation.  By merging in $\psi$ the block containing the label $l$ with the block containing the label right to $l$, we obtain $\chi'$ which has the desired properties by construction.
\end{proof}

\subsection{\texorpdfstring{Proof of Lemma~\ref{lemma_for_shelling}}{Proof of Theorem 34}}

From hypotheses of Lemma~\ref{lemma_for_shelling}, it follows that $y$, $y'$ and $z$ satisfy $\gamma(y') \leq \gamma(y) \leq \gamma(z)$. We decompose the proof into several steps: the first three steps deals with equality cases in the previous relation. From Case $4$, we will presume that $\gamma(x) \leq \gamma(y') < \gamma(y) < \gamma(z)$

\subsubsection{Case 1: $\gamma(y') = \gamma(y) = \gamma(z)$}

By Lemma~\ref{lemm:veeinpp}, $\gamma(y') = \gamma(z)$ is also the code of $y'\vee z$.  So all element in the $[x , y' \vee z]$ have the same code, except possibly $x$.  It follows that the natural projection $\pp_n \to NC_n$ sends $[x,y' \vee z]$ to an interval in $NC_n$, in a way which is compatible with the orders $\prec_x$, $\prec_y$, etc.  So this case follows from Lemma~\ref{lemma:shelling_nc}.

\subsubsection{Case 2: $\gamma(y') < \gamma(y)=\gamma(z)$}

Note that we have $\gamma(x)<\gamma(y)$, since $\gamma(x) \leq \gamma(y') < \gamma(y)$.  It follows that $m(x,y)>0$ and $m(y,z)=0$.  By Lemma~\ref{lemm:increasingm}, there exists $y''$ such that $x\lessdot y''\lessdot z$ and $\gamma(y'') < \gamma(y)$.  In particular, $y'' \prec_x y$.

\subsubsection{Case 3: $\gamma(y')=\gamma(y) < \gamma(z)$}

As $y'\neq y$ and they have the same rank, we have $y<y'\vee y$.  So there exists $z'$ such that $y\lessdot z' \leq y'\vee y$.  Note that $z' \leq y'\vee y \leq y'\vee z$ .  By Lemma~\ref{lemm:veeinpp}, we have $\gamma(y)\leq \gamma(z') \leq \gamma(y'\vee y) = \gamma(y)$.  So $\gamma(z') = \gamma(y) < \gamma(z)$, and $z' \prec_y z$.

\subsubsection{Case 4: $y' \vee z = \hat 1$}

Since Cases 1 and 2 contain all situations where $\gamma(y)=\gamma(z)$, we can assume $\gamma(y)<\gamma(z)$.  If $y=(\pi,\sigma)$,
take $z' =(\pi',\sigma)$, where $\pi'\in NC_n$ is such that $\pi\lessdot \pi'$.  We clearly have $y\lessdot z'$ and $\gamma(y)=\gamma(z')$.  Since $\gamma(z') = \gamma(y) < \gamma(z)$, we have $z' \prec_y z$.  Since $y' \vee z = \hat 1$, we have $z'\leq y'\vee z$.

\subsubsection{Case 5: $p_0(y) > p_0(z)$}

Since Case~4 is already ruled out, we can assume $y'\vee z < \hat 1$, and in particular $y'\vee y < \hat 1$ (since $y'\vee y \leq y'\vee z$).  We have $p_0(y')\geq p_0(y)$ since $\gamma(y') \leq \gamma(y)$, and by Lemma~\ref{lemm:samep0} it follows $p_0(y'\vee y) = \min( p_0(y'), p_0(y) ) = p_0(y) $. 

Let $z' \in \pp_n$ such that $y\lessdot z' \leq y'\vee y$.  Note that $z'\leq y'\vee z$ since $y'\vee y \leq y'\vee z$.  We have $p_0(y) \geq p_0(z') \geq p_0(y'\vee y) = p_0(y)$, so $p_0(z') = p_0(y) > p_0(z) $ and $p_0(z') > p_0(z) $.  It follows that $\gamma(z') < \gamma(z)$, and $z' \prec_y z$.

\subsubsection{Case 6: $N(x,y') \neq N(x,y) $}

We can assume that $\gamma(x) < \gamma(y)$, otherwise $\gamma(x) = \gamma(y') = \gamma(y)$, and the situations where $\gamma(y') = \gamma(y)$ were already settled in Cases~1 and 3.  It follows that $m(x,y)>0$, and from Lemma~\ref{lemm:interval_of_4} we get $m(x,y') \neq m(x,y)$.  More precisely, since $y' \prec_x y$ we have $m(x,y') < m(x,y)$.

If there exists $y''\in \pp_n$ such that $x\lessdot y'' \lessdot z$ and $y''\prec_x y$, this case is settled.  Assume otherwise, so that we  have $m(x,y) \leq m(y,z)$ by Lemma~\ref{lemm:increasingm}.  

We thus have $m(x,y') < m(x,y) \leq m(y,z)$, so that $m(x,y') < m(y,z)$.  Besides, $m(x,y') = m(y,y\vee y')$ by Lemma~\ref{lemm:interval_of_4}.  We thus have $m(y,y\vee y') < m(y,z)$.  So, with $z' = y'\vee y$ we have $\gamma(z')<\gamma(z) $, and $z' \prec_x z$.

\subsubsection{Case 7: $N(x,y) = N(x,y') $ and $N(y,z) \not \subset N(x,y)$}

It is straightforward to see that the non-inclusion $N(y,z) \not \subset N(x,y)$ means: $z=y \vee \bar y$ where $x \lessdot \bar y$ and $N(x,y) \neq N(x,\bar y)$.  We are thus in the situation described in Lemma~\ref{lemm:interval_of_4}.  We can assume $\gamma(x) < \gamma(y)$ (see the previous case), and it follows that $m(x,y) \neq m(x, \bar y)$.

If $\gamma(\bar y) < \gamma(y)$, we can take $y'' = \bar y$ and we have $x\lessdot y''\lessdot z$, $y'' \prec_x y$ so that this case is settled.  Assume otherwise, so that $m(x,y) < m(x, \bar y)$.  By Lemma~\ref{lemm:interval_of_4}, we have $m(x, \bar y) = m(y,z)$.  So we have $m(x,y) < m(y,z)$. 

Let $z'$ be such that $y\lessdot z' \leq y'\vee y$.  By Lemma~\ref{lemm:case7}, we have $m(y,z') \leq m(x,y)$.  We thus have $m(y,z') \leq m(x,y) < m(y,z)$, and $m(y,z') < m(y,z)$.  It follows that $\gamma(z') < \gamma(z)$ and $z'\prec_y z$.

\subsubsection{Case 8 (last case)}

As Cases~6 and 7 are ruled out, we assume $N(x,y) = N(x,y') $ and $N(y,z) \subset N(x,y)$.  It follows that only one block of $x$ is involved, i.e., all  elements in $[x,y'\vee z]$ are obtained from $x$ by splitting this block.  We assume $x = \hat 0$, as we can focus on this case by ignoring the other blocks (it is easily seen that discarding the other blocks is compatible with the orders $\prec_\phi$). 

As Case~5 is ruled out, we assume $p_0(y) = p_0(z)$, and denote this quantity $n-l$ with $l\in\llbracket 1;n\rrbracket$ ($l>0$ since we can assume $\gamma(x)<\gamma(y)$, as was done in Case~6).   This means $c_l(y)$ (resp. $c_l(z)$) is the first non-zero element of $\gamma(y)$ (resp. $\gamma(z)$).  In particular, the label $l$ has a strictly smaller label to its right in $y$ (resp. $z$).

If $c_l(y) = c_l(z)$, we define $y''$ from $z$ by merging the block containing the label $l$ with the block containing the label to the right of $l$.  We have $\gamma(y'') < \gamma(y)$ as $p_0(y'') \geq p_0(z) = p_0(y) = n-l$ and $c_l(y'') < c_l(z) = c_l(y)$, so $y'' \prec_x y$ and this case is settled.  We thus assume $c_l(y) \neq c_l(z)$, so that $c_l(y) < c_l(z)$.

Now, let $b_1 = N(y,z)$, and let $b_2$ denote the other block of $y$ (since $y$ has rank 1, it has only two blocks).  We claim that $b_2$ is not a block of $y' \vee z$.  Indeed, otherwise it would also be a block of $y'$, so that $y'$ and $y$ have the same underlying noncrossing partition, and only their codes differ.  This is impossible, because it would imply $y' \vee y = \hat 1$ (and this is already ruled out by Case~4).  

Since $b_2$ is not a block of $y' \vee z$, there exists $z'$ such that $y \lessdot z' \leq y' \vee z$ and $N(y,z') = b_2$.  Note that $p_0(z') \geq p_0(y' \vee z) = p_0(z) = n-l$.  

Besides, the position of $l$ in $z'$ is the same as in $y$, since $l$ is not the label in $y$ of an element of $N(y,z') = b_2$. Indeed, since $c_l(y) < c_l(z)$, $l$ is the label in $y$ of some element in $N(y,z)=b_1$.  We thus get $c_l(z') = c_l(y)$.

From $p_0(z') \geq n-l$ and $c_l(z') = c_l(y) < c_l(z)$, we obtain $\gamma(z') < \gamma(z)$.  This completes the last case.

\section{Enumeration of chains of parking functions}
\label{sec_enum}

Edelman proved in \cite{edelman} that the zeta polynomial of $\pp_n$ is:
\[
  Z(\pp_n,k+1)  = (nk+1)^{n-1}. 
\]
In particular, setting $k=1$ we see that noncrossing 2-partitions and parking functions are equienumerous. 
  
Another result from~\cite{edelman} is that for $0\leq k \leq n-1$, the number of elements of rank $\ell$ in $\pp_n$, called the $\ell$th {\it Whitney number of the second kind}, is
\begin{equation}  \label{edelman_whit2}
   W_\ell(\pp_n) = \ell! \binom{n}{\ell} S_2(n,\ell+1), 
\end{equation}
where $S_2(n,k)$ are the Stirling numbers of the second kind.

Our main motivation for counting chains in the parking function poset comes from the relation with the homology of the poset, as explained Section~\ref{sec_homology}.  We will get in particular a nice formula for {\it Whitney number of the first kind}.  This will follow from Corollary~\ref{coro_enumchains}.

\subsection{Species and generating functions}

\begin{prop}
The species $\mathcal{P}_f$ of (non-empty) parking trees satisfies:
\begin{equation}
\mathcal{P}_f = \sum_{k \geq 1} \mathcal{E}_k \times (1+\mathcal{P}_f)^k,
\end{equation}
where $\mathcal{E}_k(V) = \delta_{|V|=k} \mathbb{K}$ (with $\mathbb{K}$ the ground field) and the species of non-empty sets is $\mathcal{E}^+  := \mathcal{E}-1=\sum_{k\geq 1} \mathcal{E}_k$.
\end{prop}

This is obtained from the tree structure, and accordingly we can write an equation in terms of symmetric functions for the Frobenius image of the characters of $\mathcal{P}_f$.

\begin{rema} This equation cannot be simplified, as for generating series, in 
\begin{equation*}
\mathcal{P}_f = \mathcal{E}^+ \circ ( \mathcal{X} \times \left( 1+\mathcal{P}_f \right)),    
\end{equation*}
where $\mathcal{E}^+$ is the species of non-empty sets and $\mathcal{X}$ is the singleton species. Indeed, this equation is not true when considering the action of the symmetric group usually defined on parking function. The action of the symmetric group associated with this equation would not only exchange letter in the associated word but also move values. For instance, the parking function $112$ would be send by the transposition $(12)$ on the parking function $113$.
\end{rema}

The set of \emph{weak $k$-chains} of parking functions on $I$ is
the set $\operatorname{PF}^I_k$ of $k$-tuples $(a_1,\ldots, a_k)$ 
where $a_i$ are parking functions on $I$ and $a_i \leq a_{i+1}$. 
The species which associates to any set $I$ the set 
$\operatorname{PF}^I_k$ is denoted by $\Ckl$.

\begin{thm} \label{thmspecies} We have:
\begin{equation}
\Ckl = \sum_{p \geq 1} \Ckmupl \times \left( t\Ckl +1\right)^p,
\end{equation}
where $\Ckmupl (V) = \delta_{|V|=p} \Ckmul (V)$ on any set $V$ of size $p$.

In terms of generating functions, this translates to:
\begin{equation} 
C_{k,t}^l = C_{k-1,t}^l \circ  \left( x \left( tC_{k,t}^l+1\right) \right).
\end{equation}
\end{thm}

\begin{proof}
  This decomposition is obtained by separating the root, of size $p$, and elements in the chain obtained from its splitting on one side and the subtree attached to its root in the minimal element of the chain and elements obtained from them on the other side. This is made possible by the fact that the splittings of the root and of its descendants do not mix. The chain obtained by restricting to the root and its splitting is equivalent to a chain of length $k-1$ as the minimal element of the chain can easily be recovered by merging all the vertices.
\end{proof}

Note that from the functional equation in terms of species, it is theoretically possible to find a formula for the character of $\mathfrak{S}_n$ acting on the chains as above.  Here we only consider the enumerative result.

\begin{rema}
In terms of usual generating series, the computations in terms of generating series are the same as if we considered chains in a poset of forests of rooted non planar trees. In such a poset, the corresponding species would satisfy the following equation, denoting by $\mathcal{F}_{k,t}^l$ the species of large chains:
\begin{equation}
    \mathcal{F}^l_{k,t} = \left(\mathcal{E}-1\right) \circ (X \left(t\mathcal{F}_{k,t}^l+1\right)^k)
\end{equation}
Finding such an order is however still an open question.
\end{rema}

From Theorem \ref{thmspecies}, we show by induction the following formula, for any $1 \leq i \leq k$, which leads to Corollary \ref{leCor}:
\begin{equation} 
C_{k,t}^l = C_{k-i,t}^l \circ  \left( x \left( tC_{k,t}^l+1\right)^i \right).
\end{equation}

\begin{coro} \label{leCor} The generating function of weak $k$-chains in the 2-partition posets satisfies:
\begin{equation} \label{eqfuncC}
C_{k,t}^l = \exp \left( x \left( tC_{k,t}^l+1\right)^k \right) -1.
\end{equation}
\end{coro}

From \eqref{eqfuncC},  $C_{k,t}^l$ is the compositional inverse of $\operatorname{ln}(1+x)\left(1+tx \right)^{-k}$.
By using Lagrange inversion, it is possible to extract the coefficients and we get:

\begin{coro} \label{coro_enumchains}
 The number of chains $\phi_1 \leq \dots \leq \phi_k$ in $\pp_n$ where $\rk(\phi_k)=\ell$ is:
 \begin{equation}  \label{eq_enumchains}
    \ell !  \binom{kn}{\ell} S_2(n,\ell+1).
 \end{equation}
\end{coro}

\subsection{Bijective proof of Corollary~\ref{coro_enumchains}} \label{ktree}
We give here a bijective proof of this corollary, relying on the notion of $k$-parking trees.

\begin{defi}
A \emph{$k$-parking tree} on a set $L$  is a rooted plane tree $T$ such that:
\begin{itemize}
    \item internal vertices of $T$ are labelled with nonempty subsets of $L$, which form a set partition of $L$,
    \item leaves are labelled by empty sets,
\item each vertex has as many children as $k$ times the number of elements in its label.
\end{itemize}
\end{defi}

\begin{figure}
    \centering
\includegraphics{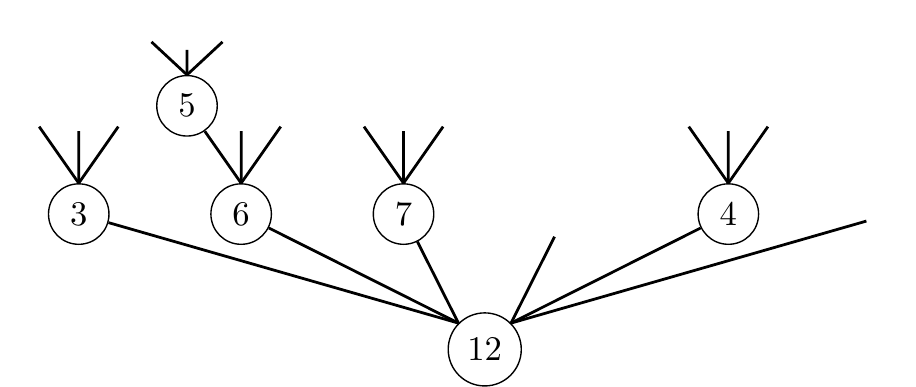}

\includegraphics[scale=0.7]{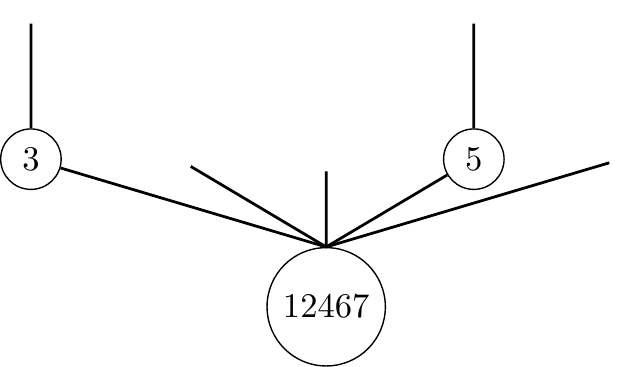} $\leq$ \includegraphics[scale=0.7]{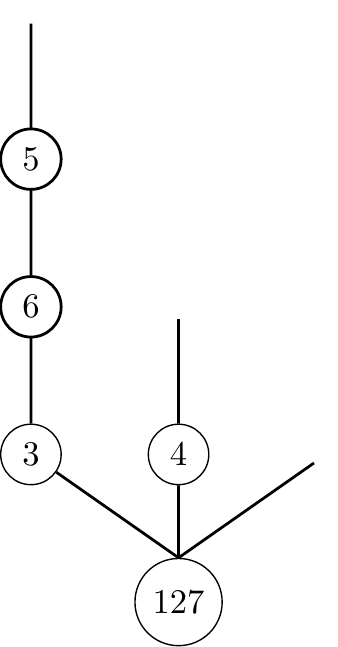} $\leq$ \includegraphics[scale=0.7]{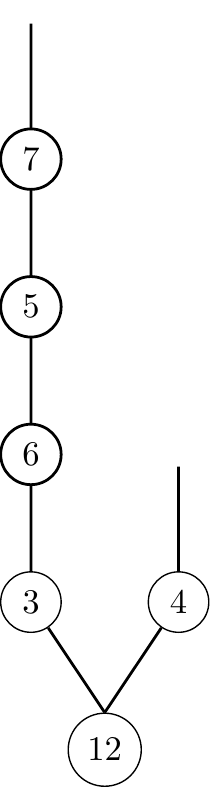}
    \caption{A $3$-parking tree and the associated chain in the poset}
    \label{fig:kTree}
\end{figure}

We will group edges from a vertex into uplets of $k$ edges, that will be called \emph{broods}, as drawn on Figure \ref{fig:kTree}. The $i$th element of this uplet will be called the child of index $i$. On the example on top of Figure~\ref{fig:kTree}, $7$ is then the child of index $3$ of the first brood of $12$, whereas $6$ and $4$ are children of index $2$, respectively of the first and second broods.

Yan~\cite{yan} introduced a variant of parking functions, such that there are $(kn+1)^{n-1}$ of them.  They can also be related with $k$-element chains in $\pp_n$.

\begin{defi}
  A $k$-parking function  of length $n$ is a word $w_1\dots w_n$ of positive integers, such that for all $j\in\llbracket1;n\rrbracket$, $\#\{ \, i \, : \, w_i \leq k(j-1)+1 \} \geq j$ (equivalently, the increasing sort of $w_1\dots w_n$ is below $1,k+1,2k+1,\dots$, entrywise). The symmetric group acts on $k$-parking functions in a natural way: for $\sigma\in\mathfrak{S}_n$, $\sigma\cdot (w_1 \dots w_n) = w_{\sigma^{-1}(1)} \dots w_{\sigma^{-1}(n)}$.
\end{defi}

There is between $k$-parking trees and $k$-parking functions of the same kind as the bijection between parking trees and parking functions stated in Lemma~\ref{BijParkFunctTree}. 

\begin{lemm}
There is a bijection between $k$-parking trees on $\llbracket 1 ; n \rrbracket$ and $k$-parking functions of length $n$, which preserves the action of the symmetric group.
\end{lemm}

\begin{proof}
The proof of Lemma \ref{BijParkFunctTree} can be adapted to the $k$ case. The condition for a vertex labelled by $E_i$ to have $k|E_i|$ children is equivalent to the condition for a weak set composition to correspond to a $k$-parking function. This condition is given by 
\begin{equation}\label{condK}
\sum_{j=1}^i k |E_j| \geq i.
\end{equation}

Note that the set compositions considered here is a weak set composition of $\llbracket 1 ; n 
\rrbracket$ of length $kn+1$.
It is direct that a set composition read from a $k$-parking tree by reading the labels in a 
prefix manner satisfies Equation \eqref{condK}. Moreover, all the reasoning of the converse 
can be adapted to this case. The most tricky point of it is perhaps ensuring that for all 
inner vertices, the cardinality of its label coincides with the number of its children. Once 
again, the grafting algorithm ensures that an internal vertex does not have more children than
$k$ times the number of elements in its label. Denoting by $c_i$ the number of children of 
vertex $E_i$, we have $c_i \leq k|E_i|$. Moreover, every $E_i$ are grafted so the tree has 
$kn+1$ nodes, i.e. $kn+1=\sum_{i=1}^{n+1} c_i +1 \leq \sum_{i=1}^{n+1} k| E_i| +1\leq kn+1$. 
Hence every inequalities are equalities and every $E_i$ has exactly $k|E_i|$ children. 
\end{proof}

We can now prove Corollary \ref{coro_enumchains}, which is illustrated on Figure \ref{fig:kTree}.

The proof of Corollary \ref{coro_enumchains} will be in two steps : first proving that such chains are in bijection with $k$-parking trees having $\ell+1$ non-empty nodes and then that they are enumerated by the formula. The following corollary immediately follows from the bijection between $k$-chains in the poset and $k$-parking trees. 
\begin{coro}
  Relations in the poset are given by $2$-parking trees. 
  In more details, for any $a$ smaller than $b$ in the poset, there exists a $2$-parking tree $T$ such that 
  \begin{itemize}
      \item $a$ is obtained by merging children of index $2$ with their parent
      \item and $b$ is obtained by grafting every child of index $2$ in $T$ on the rightmost leaf of its elder sibling.
  \end{itemize}
    In other words, a parking tree $a$ is smaller than a parking tree $b$ in the poset if and only if:
  \begin{itemize}
      \item the nodes of $a$ are obtained as unions of some nodes of $b$
      \item $a$ is obtained from $b$ by choosing a set of rightmost edges in $b$ and for each edge $e$ in it, between a parent $p$ and a child $c$, by
      \begin{enumerate}
          \item deleting $e$
          \item merging $c$ with $p$ or any of its ancestor for which $e$ is on the rightmost branch of one of its child.
      \end{enumerate}
  \end{itemize}
 \end{coro}

\begin{proof}[Proof of Corollary \ref{coro_enumchains}, step 1]

Let us first explain the bijection $\Phi$ between chains $\phi_1 \leq \dots \leq \phi_k$ in $\pp_n$ where $\rk(\phi_k)=\ell$ and $k$-parking trees having $\ell+1$ non-empty nodes.
Consider a $k$-parking tree having $\ell+1$ non-empty nodes. The $i$th parking tree of the associated chain is obtained by merging every children of indices strictly more than $i$ with their parent. In a same brood, the subtree whose root is the child of index $j$, for $2 \leq j \leq i$ is then grafted on:
\begin{itemize}
    \item the rightmost leaf of the subtree whose root is the child of index $j-h$ if this subtree is not empty and the children of indices $j-h+1$, ..., $j-1$ are empty,
    \item its parent otherwise.
\end{itemize}
We denote by $\Phi(T)[\ell]$ the $\ell$th tree of the chain $\Phi(T)$.
We do not create any cycle and transform every brood into a unique 
child: we then get a parking tree. Moreover, from a tree to 
another, the only differences are labels split with half of it 
brought to the rightmost leaf of the other: this exactly corresponds 
to covering relations in the poset. The obtained uplet of trees is 
then a chain of length $k$ in the parking poset.

To exhibit the inverse bijection $\Psi$, we start from a chain $\phi_1 \leq 
\dots \leq \phi_k$ in $\pp_n$. Let us construct a $k$-parking tree 
$T$ from it by induction on the number of vertices in $\phi_k$.
If $\phi_k$ has only one vertex $R$, the chain is constant ($\phi_i=\phi_k$, for every $i$). The associated $k$-parking tree is the $k$-parking tree with only one node labelled by $R$.  
Otherwise, let us assume that $\phi_k$ has $N+1$ vertices.
The vertices of $T$ are the $N+1$ vertices of $\phi_k$. 
Starting from the root, we then construct inductively $T$. 
The root $R$ of $T$ is exactly the root of $\phi_k$. The 
$j$th subtree $S_j$ of this root in $\phi_k$ gives the $j$th
 brood of $R$. Indeed, the root of $S_j$ is the child of 
index $i$ of the brood if it splits from $R$ at time $i$. 
Then, the position of indices strictly smaller than $i$ in 
the brood are empty (otherwise, they would be occupied by a 
descendant of the root who would also be a parent of $R$). 

Let us split the rightmost branch of $S_j$ into subtrees $S_j^i, \ldots, S_j^k$. To do so, we run from $R$ to the rightmost leaf of $S_j$. We define iteratively the forest and a strictly increasing function $f : \mathbb{N}^* \rightarrow \mathbb{N}^*$. The first step is to cut the edge between $R$ and $S_j$, then $f(1)=i$ (the index of the root of $S_j$). At step $\ell$, we cut the next edge such that the child is in the same vertex as $R$ in $\phi_{f(\ell)-1}$, but not in $\phi_{f(\ell)}$, with $k \geq f(\ell)>f(\ell -1)$. We end the procedure when either we reach the rightmost leaf of $S_j$ or at step $m$ where $f(m)=k$. We then define $S_j^{f(\ell)}$ to be the subtree whose root is the child of the edge cut at step $\ell$. The other $S_j^p$ are empty. We call, in what follows, this procedure the \emph{cutting procedure}.

\begin{lemm} \label{lem:LemmakPark}
To each $S_j^\ell$ can be associated a chain of parking tree, i.e. ${\phi_1}_{|V(S_j^\ell)}<\ldots<{\phi_k}_{|V(S_j^\ell)}$ is a chain of parking trees. 
\end{lemm}

\begin{proof}[Proof of Lemma \ref{lem:LemmakPark}]
  The only thing to prove is that ${\phi_i}_{|V(S_j^\ell)}$ is connected for every $i$. Let us consider a node $D$ in  $S_j^\ell$, and $C$ the root of $S_j^\ell$. We want to prove that in every $\phi_i$, it is either in the root of $\phi_i$ and $i<\ell$ or in a node which is a union of some nodes in  $S_j^\ell$ (excluding any other node).  
  
  Suppose first that $D$ is in the root of $\phi_1$. There is one $t$ such that $D$ is in the root of $\phi_{t-1}$ but no more in the root of $\phi_t$ as $D$ is not in the root of $\phi_k$. As $D$ is in $S_j^\ell$, $t$ is smaller or equal to $\ell$, otherwise it would have been split from the subtree $S_j^\ell$ by the cutting procedure. 
  
  Let us proceed now by the absurd and suppose that there is $B$ in $S_j^k$ ($k<\ell$) such that $B$ and $D$ are in the same vertex in $\phi_t$, this vertex being different from the root. If $t$ was smaller than $\ell$, the root of $S_j^\ell$ could not be an ancestor of $D$ because one cannot insert the splitting of a node on a path between two nodes. Then $t$ is greater than or equal to $\ell$.
  At step $t-1$, the vertex $D$ can only split to a position in one of the descendant subtree of the node $bd$ containing both $B$ and $D$. Hence, the only way for $D$ to be in $S_j^\ell$ is that $C$ is a descendant of $bd$ and $D$ splits at step $t$ to the subtree rooted in $C$. This splitting is only allowed if $C$ is not on the rightmost branch of this subtree. However, as $C$ splitted from the root after $k$, it can only be on the rightmost branch of this subtree: we get a contradiction. Hence, every $S_j^k$ evolves independently. 
  
  This proof is illustrated on Figure \ref{fig:PreuveLem}.
  \begin{figure}
      \centering
      $\phi_{l-1}=$~\includegraphics[scale=0.8]{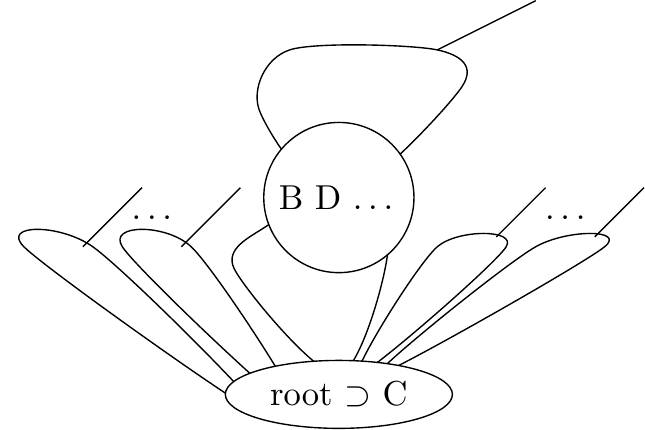},
      $\phi_{t-1}=$~\includegraphics[scale=0.8]{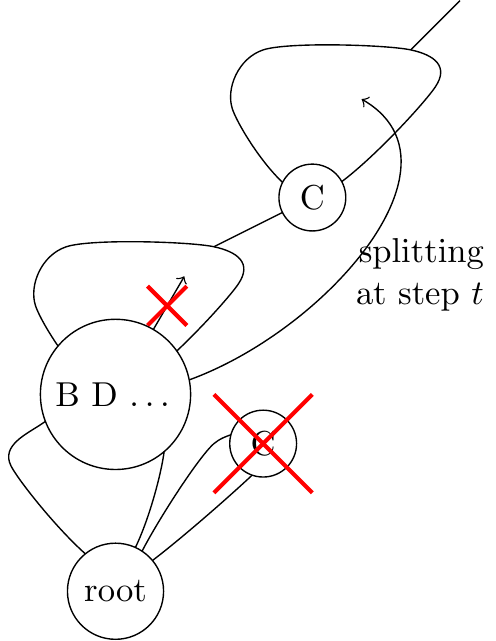},
      $\phi_{k}=$~\includegraphics[scale=0.8]{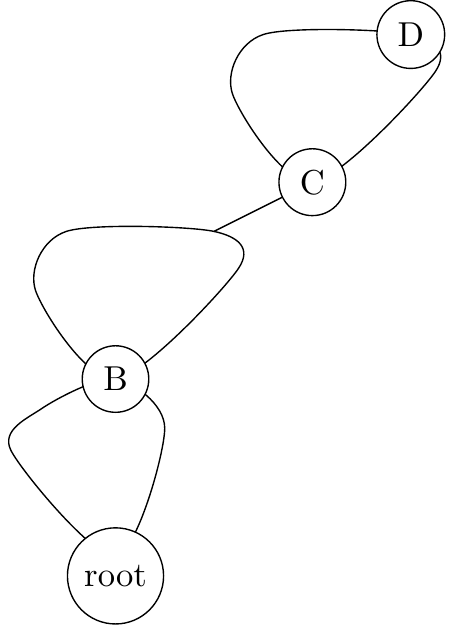}
        where    \includegraphics{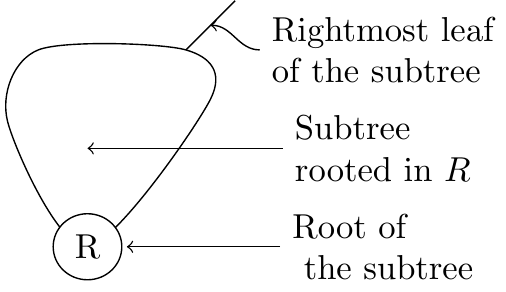}
      \caption{From top to bottom and left to right, $\phi_{l-1}$, $\phi_{t-1}$, $\phi_{k}$ in the proof of Lemma \ref{lem:LemmakPark} and legend of the drawings.}
      \label{fig:PreuveLem}
  \end{figure}
\end{proof}

The root of the subtree $S_j^\ell$, if it exists, is then defined to be the child of index $\ell$ of the $j$th brood of the root of $T$. Otherwise, this index of the brood correspond to an empty tree. Using Lemma \ref{lem:LemmakPark}, we can then apply the induction on each subtrees $S_j^i, \ldots, S_j^k$, for every $j$, to get the full $k$-parking tree $T$. 
\end{proof}

\begin{proof}[Proof of Corollary \ref{coro_enumchains}, step 2]


We now enumerate $k$-parking trees.  Consider the set of $k$-parking trees on $\llbracket 1;n\rrbracket$ having $\ell+1$ non-empty nodes. We can associate to it a kind of Prüfer code. First, we consider that a set $A$ is smaller than a set $B$ if the minimal element of $A$ is smaller than the minimal element of $B$. Let us consider a $k$-parking tree $T$ and number the half edges on the parent side from $1$ to $\ell$. Starting from $T_0=T$ at step $0$, we iteratively delete the smallest leaf $l_k$ of the remaining tree $T_k$, obtaine dafter the kth iteration of the loop, until no vertex is left in $T_\ell$. To each $l_k$ can be associated the number of its associated edge : we get a permutation $l_0\ldots l_{\ell-1}$ of $\{1, \ldots, \ell\}$.  We describe this algorithm in pseudo-code below, with $\epsilon$ denoting the empty word:

\SetCustomAlgoRuledWidth{13.2cm} 
\begin{algorithm}[H]
\SetAlgoLined
\KwResult{permutation $w$ of $\{1, \ldots, \ell\}$}
$T$ initial $k$-parking tree \;
$w \gets \epsilon$ \;
\For{$i\gets0$ \KwTo $\ell-1$ }{
$l \gets $ smallest leaf of $T$ \;
$w \gets w+$number of the half-edges attached to $l$ \;
$T \gets$ tree obtained by deleting $l$ in $T$\;
}
\Return w \;
 \caption{Construction of the code associated to a $k$-parking tree}
\end{algorithm}

Conversely, from a choice of vertex set ($S_2(n,\ell+1)$) $V$ (with $k|v|$ half-edges attached to each vertex $v$), a choice of used half-edges($\binom{kn}{\ell}$), and a permutation $\sigma$ of $\{1, \ldots, \ell\}$ ($\ell!$), one can build back a $k$-parking tree by iterating the following algorithm. We initialize the set $L$ with all vertices whose half-edges are not numbered by an integer between $1$ and $\ell$ and the word $w$ by $\sigma$. As long as $w$ is not empty, we pop (choose and remove) the element of $L$ with the smallest root and attach it to the half-edge corresponding to the first letter in $\sigma$. We get a new tree $t$. We delete the first letter of $\sigma$ and add $t$ in $L$ if no half-edge of it is labelled by an element of $\sigma$. We describe this algorithm in pseudo-code below:
\begin{algorithm}[H]
\SetAlgoLined
\KwResult{$k$-parking tree}
 $L \gets$ vertices of $V$ with no half-edges labelled\;
 $w \gets \sigma$\;
 \While{$w \leq \epsilon$}{
  $t_0 \gets \min L$\;
  $L \gets L-\{t\}$\;
  $t \gets $ Grafting of $t_0$ on the tree with half-edge $w[0]$\;
  $w \gets w[2:]$\;
  \If{$t$ has no half-edge in $w$}{
   $L \gets t$ \;
   }
  \If{$\operatorname{length}(w)=0$}{
   \Return $ t$ \;
   }
 }
 \caption{Construction of a $k$-parking tree}
\end{algorithm}
The termination of the algorithm is given by the strict decrease of the length of $w$. $L$ is never empty because there are $\ell+1-k$ trees at the $k$th iteration and $\ell-k$ numbered half-edges.
The two algorithms are reciprocal. This shows that a $k$-parking tree is equivalent to a partition, with elements of the partition having $k$ half-edges, among which we choose $\ell$ half-edges and use a permutation to encode the grafting on this half-edges. Hence they are counted by $\ell ! \binom{kn}{\ell} S_2(n, \ell+1)$.
\end{proof}

Clearly, the formula in \eqref{eq_enumchains} specializes to \eqref{edelman_whit2}, by letting $k=1$.
Also, using a general fact linking the zeta polynomial of a poset with its Möbius function, at $k=-1$ the formula above specializes to the {\it Whitney numbers of the first kind} of $\pp_n$, defined by:
\[
  w_{\ell}(\pp_n) = \sum_{\phi \in\ppp_n ,  \, \rk(\phi)=\ell }  \mu( \hat0 , \phi).
\] 
Note that the number $\mu( \hat0 , \phi)$ is a product of Catalan numbers.  Indeed, this interval is isomorphic to an interval in $NC_n$, so it follows from the result on the Möbius function of $NC_n$~\cite{kreweras}.  By letting $k=-1$ in \eqref{eq_enumchains}, we get 
\[
  w_{\ell}(\pp_n) = (-1)^{\ell} \ell! \binom{n+\ell-1}{n} S_2(n,\ell+1).
\]
In general, Whitney numbers of the first kind are the dimensions of the {\it Whitney modules}, which are useful to compute the homology of a poset (see the definition in the next section).

\section{Homology of the parking function poset}

\label{sec_homology}

We now study the homology associated to the parking function poset.  The reader may read Wachs' article \cite{wachs} as a general reference on this subject (in particular for Philip Hall's theorem, the Hopf trace formula, Whitney homology), and Munkres' book \cite{munkres} for more details on simplicial homology. 

\subsection{A first derivation using the zeta character}

Let $\bar\pp_n$ denote the {\it proper part} of $\pp_n$, i.e., $\pp_n$ with its bottom element removed (the topology associated to $\pp_n$ is trivial, so $\bar\pp_n$ is the poset to consider here).  We denote by $\Omega(\bar\pp_n)$ the {\it order complex} of $\bar\pp_n$, i.e., the simplicial complex having strict chains in $\bar\pp_n$ as simplices.  We are interested in the reduced simplicial homology of $\Omega(\bar\pp_n)$, but let us be more explicit.

\begin{defi}
For $-1\leq m\leq n-2$, the \emph{$m$th space of chains} is the vector space $\mathcal{C}_m$ freely generated by $m$-dimensional simplices in $\Omega(\bar\pp_n)$ (i.e., strict chains $\phi_1<\dots<\phi_{m+1}$, where $\phi_i\in\bar\pp_n$).  For $0\leq m \leq n-2$, we define a linear map $\partial_m:\mathcal{C}_m\to \mathcal{C}_{m-1}$ as follows: if $\Delta=\{\phi_1,\dots,\phi_{m+1}\}\in\Omega(\bar\pp_n)$ with $\phi_1<\dots<\phi_{m+1}$, then
\[
  \partial_m( \Delta ) = \sum_{i=1}^{m+1} (-1)^i \cdot (\Delta\setminus\{\phi_i\}).
\]
It is straightforward to check that $\partial_{m}\circ\partial_{m+1}=0$.  For $-1\leq m \leq n-2$, the $m$th reduced homology space of $\bar\pp_n$ is $\tilde H_m(\bar\pp_n) = \ker \partial_{m} / \im \partial_{m+1}$. (By convention, $\ker \partial_{-1} = \im \partial_{n-1} = \{0\}$.)
\end{defi}

Note that the action of $\mathfrak{S}_n$ on chains in $\bar\pp_n$ permits us to view $\mathcal{C}_m$ as a $\mathfrak{S}_n$-module.  It is clear that the maps $\partial_m$ are module maps, so that $\tilde H_m(\bar\pp_n)$ is also a $\mathfrak{S}_n$-module.  

As a consequence of the shellability property obtained in Theorem~\ref{theo_shelling}, $\Omega(\bar\pp_n)$ has the homotopy type of a bouquet of $n-2$-dimensional spheres, so $\dim \tilde H_m(\bar\pp_n)=0$ for $m \neq n-2$. 

\begin{thm} \label{theo_charhomology}
The character of $\tilde H_{n-2}(\bar\pp_n)$ as a representation of $\mathfrak{S}_n$ is given by:
\begin{align}  \label{formula_charhomology}
  \sigma \mapsto (-1)^{n-z(\sigma)} (n-1)^{z(\sigma)-1}.
\end{align}
\end{thm}

\begin{proof}
  We can use the result in \cite[Proposition~1.7]{sppp}, and it follows that the desired character is $(-1)^{n-1}$ times the specialization at $k=-1$ of \eqref{char_chains}.  This gives the desired formula.
  
  It's worth writing that more explicitely.  First, the {\it Hopf trace formula} gives the equality
  \begin{equation} \label{hopftrace}
    \sum_{i=-1}^{n-2} (-1)^i \mathcal{C}_i = \sum_{i=-1}^{n-2} (-1)^i \tilde H_i ( \bar\pp_n )  
  \end{equation}
  in the representation ring of $\mathfrak{S}_n$.  Since only one term is nonzero in the right-hand side, this is also equal to $(-1)^n \tilde H_{n-2}(\bar\pp_n)$.
  
  Let $\mathcal{D}_k$ be the vector space freely generated by large chains $\phi_1\leq \dots \leq \phi_k$ in $\pp_n$.  It is a $\mathfrak{S}_n$-module in a natural way, and its character is $\operatorname{Park}^{(k)}_n$ (see~\eqref{char_chains}).  A large chain in $\pp_n$ can be obtained from a strict chain in $\bar\pp_n$ by choosing some multiplicities for each element in the chain, and adding the minimal element of $\pp_n$ with some multiplicity.  Omitting details, for any $k\geq 0$ this gives the relation
  \begin{align}  \label{rel_strict_large}
   \mathcal{D}_k = \sum_{i=-1}^{n-2} \binom{k}{i+1}  \mathcal{C}_i
  \end{align}
  in the representation ring of $\mathfrak{S}_n$.  By a polynomiality argument, we can set $k=-1$ in \eqref{rel_strict_large}.  What we get on the right hand side is the alternating sum in the left-hand side of~\eqref{hopftrace}, up to a sign.  Thus, we have $\mathcal{D}_{-1} = (-1)^{n-1} \tilde H_{n-2}(\bar\pp_n) $.
  So the character of $H_{n-2}(\bar\pp_n)$ is $(-1)^{n-1} \operatorname{Park}^{(-1)}$.  This gives $\sigma \mapsto (-1)^{n-1} (1-n)^{z(\sigma)-1}$, and we get~\eqref{formula_charhomology}.
\end{proof}

\begin{coro}
The Möbius number of $\hat\pp_n$ is $\mu ( \hat\pp_n ) = (-1)^n (n-1)^{n-1}$.
\end{coro}

\begin{proof}
  By Philip Hall's theorem, $\mu( \hat\pp_n )$ is the Euler characteristic of $\Omega(\bar\pp_n)$. So it is also $(-1)^n\dim \tilde H_{n-2}(\bar\pp_n)$.  This comes from taking $\sigma=id$ in~\eqref{formula_charhomology}.
\end{proof}

\subsection{A second derivation using Whitney modules}

Another method to compute the character in Theorem~\ref{theo_charhomology} consists in using {\it Whitney homology}.  For each $\pi\in NC_n \setminus\{0_n\}$, denote $(0_n,\pi)$ the open interval in $NC_n$.  If $\pi$ has rank $|\pi|-1 = \ell$, the only nonzero homology group of this open interval is $\tilde H_{\ell-2} ( (\hat 0 , \pi) )$.  It is seen as a $\mathfrak{S}_n(\pi)$-module in a trivial way, each group element acting as the identity.  By definition, $\ell$th Whitney module $\mathcal{W}_{\ell}(\pp_n)$ is defined as the sum of induced representations
\[
   \mathcal{W}_{\ell}(\pp_n) = 
   \sum_{\pi \in NC_n, \, |\pi|-1=\ell} \operatorname{Ind}_{\mathfrak{S}_n(\pi)}^{\mathfrak{S}_n} \tilde H_{\ell-2} ( (\hat 0 , \pi) )
\]
if $1\leq \ell\leq n-1$, and by convention $\mathcal{W}_{0}(\pp_n)$ is the trivial $\mathfrak{S}_n$-module.

To see that this definition coincides with that in \cite[Chapter~4]{wachs}, note that orbit representatives for rank $\ell$ elements in $\pp_n$ are the elements $(\pi,\pi,id)$ where $\pi \in NC_n$ has rank $\ell$.  The open interval $((0_n,0_n,id),(\pi,\pi,id))$ in $\pp_n$ is isomorphic to the open interval $(0_n,\pi)$ in $NC_n$ via Lemma~\ref{lemm:intervproj}.  Moreover, $\sigma\in\mathfrak{S}_n$ acts as the indentity on this interval if $\sigma \cdot (\pi,\pi,id) = (\pi,\pi,id)$  (i.e., $\sigma\in\mathfrak{S}_n(\pi)$).

\begin{lemm} \label{MoebInt}
For $\pi\in NC_n\setminus\{ 0_n \}$ of rank $\ell$, we have:
\[
  \dim \tilde H_{\ell-2} ( ( 0_n , \pi) )
  =
  \prod_{b\in K(\pi)} C_{|b|-1}.
\]
\end{lemm}

\begin{proof}
By Philipp Hall's theorem, this dimension is the Möbius number of the interval $[0_n, \pi]$ in $NC_n$, up to a sign.  Using the Kreweras complement, this interval is isomorphic to $[K(\pi),1_n]$, thus isomorphic to $NC_{i_1}\times NC_{i_2}\times \cdots$ where $i_1,i_2,\dots$ are the block sizes of $K(\pi)$.  The result follows from the fact that the Möbius number of $NC_n$ is $(-1)^{n-1}C_{n-1}$.
\end{proof}

\begin{proof}[Proof of Theorem~\ref{theo_charhomology}]
Using the previous lemma, we see that the Whitney modules of $\pp_n$ are given by:
\begin{equation}  \label{eq_whit}
  \mathcal{W}_{\ell}(\pp_n) = \sum_{\pi\in NC_n, \; |\pi|=\ell-1 } \bigg(\prod_{b\in K(\pi)} C_{|b|-1} \bigg) \operatorname{Ind}_{\mathfrak{S}_n(\pi)}^{\mathfrak{S}_n}(1)
\end{equation}
for $0\leq \ell\leq n-1$ (where $1$ denotes the trivial character of $\mathfrak{S}_n(\pi)$).  By a theorem of Sundaram (see~\cite[Theorem~4.4.1]{wachs}), the module $\tilde H_{n-2}(\bar\pp_n)$ can then be obtained as an alternating sum:
\[
  \tilde H_{n-2}(\bar\pp_n)  =  (-1)^{n-1} \sum_{\ell = 0 }^{n-1} (-1)^{\ell} \mathcal{W}_{\ell}(\pp_n).
\]
Using~\eqref{eq_whit}, we compute the right-hand side of the previous equation.  This gives:
\[
  \tilde H_{n-2}(\bar\pp_n)  
  =
  (-1)^{n-1} \sum_{\pi\in NC_n} (-1)^{|\pi|-1} 
  \bigg(\prod_{b\in K(\pi)} C_{|b|-1} \bigg) \operatorname{Ind}_{\mathfrak{S}_n(\pi)}^{\mathfrak{S}_n}(1).
\]
By polynomiality in $k$, we can plug $k=-1$ in~\eqref{eq_parkk}.  Using the reciprocity $C^{(-1)}_n = (-1)^{n-1}C_{n-1}$, we get:
\[
  \operatorname{Park}^{(-1)}_n 
  = 
  \sum_{\pi\in NC_n} \bigg( \prod_{b\in K(\pi)} 
  (-1)^{|b|-1} C_{|b|-1} \bigg) 
  \operatorname{Ind}_{\mathfrak{S}_n(\pi)}^{\mathfrak{S}_n}(1).
\]
To check the signs, note that
\[
  \sum_{b\in K(\pi)} (|b|-1) = n - |K(\pi)| = |\pi| - 1.
\]
We thus have $\tilde H_{n-2}(\bar\pp_n) = (-1)^{n-1} \operatorname{Park}^{(-1)}_n$.   As $\operatorname{Park}^{(-1)}_n(\sigma) = (1-n)^{z(\sigma)-1}$ for $\sigma\in\mathfrak{S}_n$, we obtain the formula in Theorem~\ref{theo_charhomology}.
\end{proof}

We now give a combinatorial interpretation of Lemma~\ref{MoebInt}.
\begin{defi}
  Let $T$ be a parking tree and $x$ be a vertex of $T$. The \emph{right branch} of $x$ is the set of vertices $v$ in the subtree of $T$ rooted in $x$ such that the unique path between $v$ and $x$ only contains $x$ and vertices which are the rightmost child of their parent. We denote it by $RB(T)$
\end{defi}

\begin{lemm}
For $(\pi, \sigma) \in \pp_n$ and $T$ the associated parking tree, the size of a block $b$ of the Kreweras complement $K(\pi)$ decreased by one is the number of vertices on the right branch of parking tree $T_b$ encoding the noncrossing partition under $b$.
\begin{equation}
    |b|=|RB(T_b)|+1 
\end{equation}
\end{lemm}

\begin{proof}
This lemma is a direct consequence of the bijection described in Lemma~\ref{BijParkNCP}. Indeed, the size of a block minus 1 is the number of "arches" of it. The bijection sends the non-crossing partition under the first arch to the root of the tree $T_b$ and the non-crossing partition left to the right subtree of $T_b$.
\end{proof}

\subsection{Prime parking functions}
\label{sec:prime}

In the context of the parking space theory, there is a character closely connected to the one in~\eqref{formula_charhomology}, combinatorially related to the notion of {\it prime parking functions}.

\begin{defi}
  A noncrossing partition $\pi\in NC_n$ is {\it prime} if $1$ and $n$ are in the same block of $\pi$.  An element $(\pi,\rho,\lambda)\in\pp_n$ is {\it prime} if $\pi$ is a prime noncrossing partition.  Denote by $NC'_n\subset NC_n$ the subset of prime noncrossing partitions, and $\pp'_n\subset\pp_n$ the subset of prime noncrossing 2-partitions. 
\end{defi}

Algebraically, note that $\pi\in NC_n$ is prime iff $\bar\pi$ does not belong to a proper Young subgroup of $\mathfrak{S}_n$.  As a word $w_1 \dots w_n$, a parking function is prime iff $\#\{ i \; | \; w_i \leq k\}>k$ for $k\in\llbracket1;n-1\rrbracket$. On parking trees, it corresponds for the root to have a leaf as its rightmost child. Following Section~\ref{sec_parkingspace}, the character of $\mathfrak{S}_n$ acting on $\pp'_n$ is:
\[
  \operatorname{Park}'_n := \sum_{\pi \in NC'_n }  \operatorname{Ind}_{\mathfrak{S}_n(\pi)}^{\mathfrak{S}_n}(1). 
\]

\begin{prop} \label{prop_park}
For $\sigma\in\mathfrak{S}_n$, we have:
\begin{equation} \label{parkprim}
  \operatorname{Park}'_n (\sigma) 
  =
  (n-1)^{z(\sigma)-1}.
\end{equation}
\end{prop}

\begin{proof}
This can be proved using a connection with rational parking functions of Armstrong, Loehr and Warrington~\cite{ALW}.  For two coprime positive integers $a$ and $b$, these authors define $(a,b)$-parking functions as a lattice path that stays above the diagonal in a $a\times b$-rectangle, with some labels on the up steps.  They show that parking functions (in the usual sense) correspond to the case $(a,b) = (n,n+1)$.  A similar argument shows that prime parking functions correspond to the case $(a,b) = (n,n-1)$.  As the character $\mathfrak{S}_a$ acting on $(a,b)$-parking functions is $\sigma\mapsto b^{z(\sigma)-1}$, we get the result.  We omit details.
\end{proof}

Let $\operatorname{Sign}$ denote the sign character of $\mathfrak{S}_n$, defined by $\operatorname{Sign}(\sigma) = (-1)^{n-z(\sigma)}$ for $\sigma\in\mathfrak{S}_n$.  By comparing~\eqref{formula_charhomology} and~\eqref{parkprim}, we obtain the following:

\begin{coro}
  The character of $\tilde H_{n-2}(\bar\pp_n)$ is 
  \[
  \operatorname{Char}(\tilde H_{n-2}(\bar\pp_n))
  =
  \operatorname{Sign} \otimes \operatorname{Park}'_n.
  \]
\end{coro}

It would be very interesting to have a direct proof of this equality, without an explicit computation of both sides.  This could be done by finding a basis $(e_\phi)_{\phi\in\ppp'_n}$ of $\tilde H_{n-2}(\bar\pp_n)$, such that $\sigma\cdot e_\phi = \operatorname{Sign}(\sigma) e_{\sigma\cdot\phi} $.  

Let's do that explicitly for $n=3$.  The Hasse diagram of $\bar\pp_3$ is represented in Figure~\ref{fig_hassepp3}, in a way that respects the symmetry of the graph rather than the order.  For each cycle of the underlying undirected graph, the alternating sum of its edges gives an element in $\tilde H_1(\bar\pp_3)$.  Note that each transposition $(i,j)$ acts on this graph as a reflection in the plane.  
\begin{itemize}
    \item The 12-cycle at the boundary of the picture is fixed by $\mathfrak{S}_3$.  This cycle can be matched with 111, the element of $\pp'_3$ fixed by $\mathfrak{S}_3$.
    \item Choose a length 6 cycle going through 211 (there are two of them).  The cyclic permutations of coordinates gives two other cycles, going through 121 and 112, respectively.  These three cycles can be matched with 211, 121, and 112, the three remaining elements of $\pp'_3$.
\end{itemize}
The four cycles obtained in this way define four elements in $\tilde H_1(\bar\pp_3)$.  It is straightforward to identify the action of $\mathfrak{S}_3$ on these elements.

\begin{figure}[h!tp]
    \centering
    \includegraphics[scale=0.5]{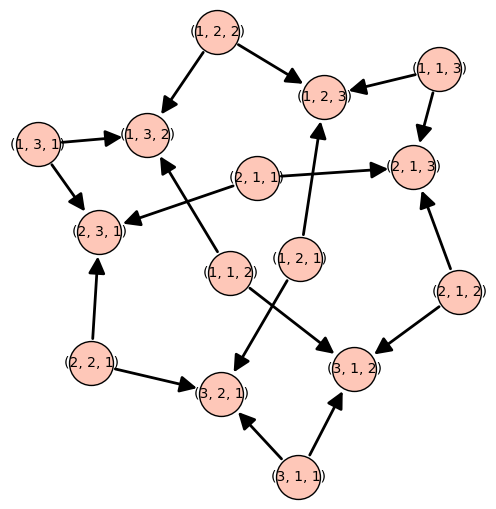}
    \caption{The Hasse diagram of $\bar\pp_3$.
    \label{fig_hassepp3}}
\end{figure}

\section{Associahedra and parking functions} \label{associahedron}

\label{sec:assoc}

The initial goal of this section was to give a combinatorial interpretation to the numbers $w_{\ell}(\pp_n)$.  This leads us to define a simplicial complex $\ppdel_n$ whose elements involve both faces of the associahedron and noncrossing $2$-partitions.  We call them {\it cluster parking function}, because the associahedron is related with the cluster complex coming from the theory of cluster algebras.  This simplicial complex $\ppdel_n$ might be useful in understanding the topology of the parking function poset.  Indeed, we will show that it has the same topology as $\pp_n$.

\subsection{The complex of noncrossing alternating forests}

The $n$th associahedron $K_n$ is a simple $n$-dimensional polytope with a long history.  For $0\leq i\leq n$, its $i$-dimensional faces are indexed by valid bracketings of $n+2$ factors with $n-i$ pairs of parentheses,  moreover the incidence relations between faces are obtained by removing or adding pairs of parentheses.  See~\cite{tamari} as a general reference.  Here, we use slightly different objects, as indicated in the title of this section.

\begin{defi}
  We denote by $\Delta_n$ the set of {\it noncrossing alternating forests} on $\llbracket 1;n\rrbracket$, i.e., forests (acyclic undirected graphs) that contains no pair of edges $\{i,j\}$ and $\{k,\ell\}$ such that $i<k\leq j < \ell$.  Moreover, let $\partial \Delta_n \subset \Delta_n$ denote the subset of forests not containing the edge $\{1,n\}$.  (It will be explained later that $\partial \Delta_n$ is the boundary of $\Delta_n$ in a precise sense.)
\end{defi}

Note that ``noncrossing'' refers to the forbidden relation $i < k < j < \ell$, which means that edges can be drawn in a noncrossing way (see example below).  And ``alternating'' refers to the forbidden relation $i < k = j < \ell$, which means that the neighbours of a vertex $i$ are all smaller or all bigger than $i$.  For example, $f=\{\{1,3\}, \{1,8\}, \{2,3\}, \{4,7\}, \{6,7\} \} \in \Delta_8$.
 
 We can identify a forest with its edge set (we always understand that $n$, hence the vertex set of the forests, is fixed once for all).  This way, $\Delta_n$ is stable under taking subsets and can be seen as a simplicial complexes such that:
\begin{itemize}
    \item its vertices are pairs $\{i,j\}$ with $1\leq i<j\leq n$, and can be identified with transpositions in $\mathfrak{S}_n$ or coatoms of $NC_n$,
    \item its facets (maximal faces) are noncrossing alternating trees. 
\end{itemize}
In particular, $\Delta_n$ is purely $n-2$-dimensional.  By taking the face poset of this simplicial complex (the set of faces ordered by inclusion), we also think of $\Delta_n$ as a poset. 
 
\begin{prop}
  The simplicial complex $\Delta_n$ is a cone over $\partial\Delta_n$.  In particular, $\Delta_n$ is topologically trivial.
\end{prop}

\begin{proof}
  It is easily checked that the forbidden relation $i<k\leq j<\ell$ cannot hold if $\{i,j\}=\{1,n\}$ or $\{k,\ell\}=\{1,n\}$.  So, for each face $f\in\partial\Delta_n$, we have $f\cup\{\{1,n\}\} \in \Delta_n$.  This means that $\Delta_n$ can be seen as a cone over its full subcomplex with vertices different from $\{1,n\}$, i.e., over $\partial \Delta_n$. 
\end{proof}
 
There is a convenient way to represent $f\in \partial \Delta_n$ as valid bracketings of $n$ factors, such as $((\bullet \bullet \bullet)  \bullet)  \bullet$, by writing a pair enclosing the $i$th and $j$th factor (and others inbetween) if $\{i,j\}\in f$.  This way, we can identify $\partial\Delta_n$ with the poset of nonempty faces of the associahedron $K_{n-2}$ ordered by reverse inclusion.  Using the dual polytope $K_{n-2}^*$, we can identify $\partial\Delta_n$ with the poset of non-maximal faces of the simplicial polytope $K_{n-2}^*$ ordered by inclusion.  The geometrical realization of $\partial\Delta_n$ can thus be identified with the boundary of $K_{n-2}^*$, i.e., a $n-3$-dimensional sphere.  It follows that the geometric realization of $\Delta_n$ is homeomorphic to a $n-2$-dimensional ball, and $\partial\Delta_n$ is indeed its boundary.
 
\begin{defi}
  For $f\in\Delta_n$, its set of connected components form a noncrossing partition that will be denoted $\underline{f} \in NC_n$.
\end{defi}

The fact that $\underline{f}$ is indeed a noncrossing partition immediately follows from the fobidden relation $i<k<j<\ell$ if $\{i,j\},\{k,\ell\}\in f$.  For example, the associated noncrossing partition associated to $f=\{\{1,3\},\allowbreak \{1,8\},\allowbreak \{2,3\},\allowbreak \{4,7\},\allowbreak \{6,7\} \} \in \Delta_8$ as above is $\underline{f} = \{\{1,2,3,8\}, \{4,6,7\}, \{5\}\}$.  It is straightforward to check that the map $f\mapsto \underline{f}$ is order-reversing.  Equivalently and using the Kreweras complement, $f\mapsto K(\underline{f})$ is order-preserving.

\begin{prop}
  We have the following relation between Whitney numbers:
  \[
    W_{\ell}(\Delta_n) 
    = 
    (-1)^{\ell} w_{\ell}(NC_n)
  \]
  for $0\leq \ell\leq n-1$.  In particular, the number of facets of $\Delta_n$ (i.e., noncrossing alternating trees on $\{1,\dots,n\}$) is the Catalan number $C_{n-1}$.
\end{prop}

\begin{proof}
A bijection between noncrossing alternating trees and complete binary trees can be given pictorially by drawing each edge $\{i,j\}$ as two line segments from $(i,0)$ to $(\frac{i+j}2,\frac{j-i}2)$ and from $(\frac{i+j}2,\frac{j-i}2)$ to $(j,0)$, see Figure~\ref{fig_alttrees}.  This proves the case $\ell=n-1$, as we get $C_{n-1}$ on both sides. 

This bijection can be extended componentwise to get the number of noncrossing alternating forests associated to a given $\pi\in NC_n$:
\begin{equation} \label{eq_deltaknc}
  \#\{ f\in \Delta_n \;:\; \underline{f} = \pi \} =  \prod_{b\in \pi} C_{|b|-1}.
\end{equation}
This product of Catalan numbers is also $(-1)^{n-1-\ell}\mu_{NC_n}(\pi,1_n)$, see Section~\ref{sec:nc}.  By summing over $\pi$ of rank $n-1-\ell$ in $NC_n$, we get
\[
  W_{\ell}(\Delta_n) 
  = \sum_{\substack{\pi\in NC_n \\ |\pi| = n-\ell}} (-1)^{\ell} \mu_{NC_n}(\pi,1_n) 
  =
  (-1)^{\ell} w_{\ell}(NC_n),
\]
where in the last equality we used the self-duality of $NC_n$.
\end{proof}

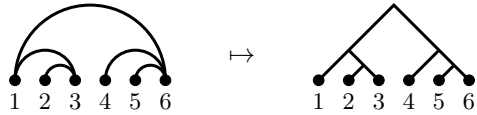
\begin{figure}[h!tp]
  \centering
  \begin{tikzpicture}[scale=0.4]
     \tikzstyle{ver} = [circle, draw, fill, inner sep=0.5mm]
     \tikzstyle{edg} = [line width=0.4mm]
     \node      at (1,-0.7) {\small 1};
     \node      at (2,-0.7) {\small 2};
     \node      at (3,-0.7) {\small 3};
     \node      at (4,-0.7) {\small 4};
     \node      at (5,-0.7) {\small 5};
     \node      at (6,-0.7) {\small 6};
     \node[ver] at (1,0) {};
     \node[ver] at (2,0) {};
     \node[ver] at (3,0) {};
     \node[ver] at (4,0) {};
     \node[ver] at (5,0) {};
     \node[ver] at (6,0) {};
     \draw[edg] (3,0) arc (0:180:1  );
     \draw[edg] (3,0) arc (0:180:0.5  );
     \draw[edg] (6,0) arc (0:180:2.5);
     \draw[edg] (6,0) arc (0:180:1  );
     \draw[edg] (6,0) arc (0:180:0.5);
  \end{tikzpicture}
  \qquad
  \begin{tikzpicture}[scale=0.4]
    \useasboundingbox (0,0) rectangle (0,2);
    \node at (0,2) {$\mapsto$};
  \end{tikzpicture}
  \qquad
  \begin{tikzpicture}[scale=0.4]
     \tikzstyle{ver} = [circle, draw, fill, inner sep=0.5mm]
     \tikzstyle{edg} = [line width=0.4mm]
     \node      at (1,-0.7) {\small 1};
     \node      at (2,-0.7) {\small 2};
     \node      at (3,-0.7) {\small 3};
     \node      at (4,-0.7) {\small 4};
     \node      at (5,-0.7) {\small 5};
     \node      at (6,-0.7) {\small 6};
     \node[ver] at (1,0) {};
     \node[ver] at (2,0) {};
     \node[ver] at (3,0) {};
     \node[ver] at (4,0) {};
     \node[ver] at (5,0) {};
     \node[ver] at (6,0) {};
     \draw[edg] (1,0) -- (3.5,2.5) -- (6,0);
     \draw[edg] (3,0) -- (2,1);
     \draw[edg] (2,0) -- (2.5,0.5);
     \draw[edg] (4,0) -- (5,1);
     \draw[edg] (5,0) -- (5.5,0.5);
  \end{tikzpicture}
  \caption{A noncrossing alternating tree and the associated complete binary tree.\label{fig_alttrees}}
\end{figure}

The poset $\Delta_n$ can be used to find what is the topology of $NC_n$.  We briefly explain this, following the results of Athanasiadis and Tzanaki~\cite{athanasiadistzanaki}.  Let $\bar {NC}_n$ (resp.~$\bar \Delta_n$) denote the poset $NC_n$ (resp.~$\Delta_n$) with its minimal element and maximal element(s) removed.  (Note that this notation with a bar is not uniform for the posets considered in this article.)  As $\Delta_n$ is topologically a $n-2$-dimensional ball, removing the $C_{n-1}$ top-dimensional simplices results in a wedge of $C_{n-1}$ many $n-3$-dimensional spheres, which is thus the topology of $\bar {\Delta}_n$.  The geometric realizations of $\bar {\Delta}_n$ and $\Omega(\bar \Delta_n)$ are homeomorphic, since the latter is the {\it barycentric subdivision} of the former.  Eventually, Athanasiadis and Tzanaki~\cite{athanasiadistzanaki} proved that the map $\Omega(\bar \Delta_n) \to \Omega( \bar {NC}_n)$ induced by $f\mapsto \underline{f}$ is a homotopy equivalence, using Quillen's fiber lemma.  It follows that $\Omega(\bar{NC}_n)$ is homotopy equivalent to a wedge of $C_{n-1}$ many $n-3$-dimensional spheres, just like $\bar \Delta_n$.

\subsection{Cluster parking functions}

By analogy with our discussion about $\Delta_n$ in the previous section, we introduce a simplicial complex $\ppdel_n$. In some sense, it is related to $\pp_n$ in the same way as $\Delta_n$ is related to $NC_n$.

\begin{defi}
  A {\it cluster parking function} is an element of the set
  \[
    \ppdel_n := \big\{ (f,(\pi,\rho,\lambda)) \in \Delta_n\times\pp_n \; \big| \; K(\underline{f})=\pi \big\}.
  \]
  A partial order on $\ppdel_n$ is defined by $(f',\phi') \leq (f,\phi)$ iff $f'\subset f$ and $\phi' \leq \phi$ in $\pp_n$.  An action of $\mathfrak{S}_n$ on $\ppdel_n$ is defined by $\sigma \cdot ( f , \phi ) = (f , \sigma\cdot\phi )$.
\end{defi}

Note that $\ppdel_n$ is a subposet of the product poset $\Delta_n\times \pp_n$, and it contains pair of elements having the same rank.  It is easily seen that the projection on each factor is a rank-preserving poset map.  Moreover the action of $\mathfrak{S}_n$ respects the order of $\ppdel_n$.

\begin{rema}
  The poset $\ppdel_n$ can be seen as the fiber product of $\Delta_n$ and $\pp_n$ over $NC_n$, along the two poset maps $\Delta_n \to NC_n, \; f\mapsto K( \underline{f} )$ and $\pp_n\to NC_n, \; (\pi,\rho,\lambda)\mapsto \pi$.  This point of view will be useful to relate the topology of the two posets $\ppdel_n$ and $\pp_n$, as we will use a fiber poset theorem.
\end{rema}

\begin{prop}
 For $0\leq\ell\leq n-1$, we have $W_{\ell}(\ppdel_n) = (-1)^{\ell} w_\ell(\pp_n)$.
\end{prop}

\begin{proof}
Using the map $(f,\phi)\mapsto \phi$ from $\ppdel_n$ to $\pp_n$, we can write:
\[
  W_{\ell}(\ppdel_n) 
  =
  \sum_{\substack{\phi\in\ppp_n \\ \rk(\phi) = \ell}}
  \#\{ f\in\Delta_n \;|\; (f,\phi)\in\ppdel_n \}
  =
  \sum_{\substack{(\pi,\rho,\lambda)\in\ppp_n \\ \rk(\pi) = \ell}}
  \#\{ f\in\Delta_n \;|\; K(\underline{f}) = \pi \}.
\]
Using~\eqref{eq_deltaknc} and the result on the Möbius function of $NC_n$, we get
\[
  W_{\ell}(\ppdel_n) 
  =
  \sum_{\substack{(\pi,\rho,\lambda)\in\ppp_n \\ \rk(\pi) = \ell}}
  \prod_{b\in K(\pi)} C_{|b|-1}
  =
  (-1)^\ell \sum_{\substack{(\pi,\rho,\lambda)\in\ppp_n \\ \rk(\phi) = \ell}}
  \mu_{NC_n}(0_n,\pi)
  .
\]
Recall that the interval $[ 0_n , \pi ]$ in $NC_n$ is isomorphic to the interval $[ (0_n,0_n,id) , (\pi,\rho,\lambda) ]$ in $\pp_n$.  The previous equation gives:
\[
  W_{\ell}(\ppdel_n) 
  =
  (-1)^\ell \sum_{\substack{(\pi,\rho,\lambda)\in\ppp_n \\ \rk(\phi) = \ell}}
  \mu_{\ppp_n}((0_n,0_n,id) , (\pi,\rho,\sigma))
  =
  (-1)^\ell w_{\ell} (\pp_n)
  .
\]
\end{proof}

\begin{prop}
  $\ppdel_n$ is the face poset of a simplicial complex.
\end{prop}

\begin{proof}
If $(f,\phi) \in \pp_n$ and $f'\subset f$, it follows from Lemma~\ref{lem:unique} that there exist unique $\rho$ and $\lambda$ such that $(K(\underline{f'}) , \rho , \lambda ) \leq \phi$ in $\pp_n$, so there exists a unique $\phi'\in\pp_n$ such that $(f',\phi') \leq (f,\phi)$ in $\ppdel_n$.  It follows that each order ideal in $\ppdel_n$ is a boolean lattice.  

Let $V$ denote the set of rank 1 element in $\ppdel_n$.  It remains only to show that the map
\begin{equation} \label{mapv}
  (f,\phi) \mapsto \{ v\in V \;:\; v\leq (f,\phi) \}
\end{equation}
is injective to identify $\ppdel_n$ with a simplicial complex having $V$ as its vertex set. 

So, let $(f,(\pi,\rho,\lambda))\in\ppdel_n$, and let $(f_i,(\pi_i,\rho_i,\lambda_i))_{1\leq i \leq k}$ be the rank 1 elements below it.  Here $k$ is the rank of $(f,(\pi,\rho,\lambda))$, since the order ideal of elements below it is boolean.  First, note that $f$ is the union of the singletons $f_i$, as if $f' \subset f$ is a singleton there exists $(f',\phi') \leq (f,\phi)$.  It follows that $\underline{f} = \underline{f_1} \wedge \dots \wedge \underline{f_k}$ in $NC_n$.  By taking the Kreweras complement, we get $\pi = \pi_1 \vee \dots \vee \pi_k$.  Eventually, we show that $(\pi,\rho,\lambda) = \vee _{1\leq i \leq k}(\pi_i,\rho_i,\lambda_i)$.  Otherwise, the join would be of rank $<k$, and taking the projection on $NC_n$ give a contradiction since $\pi = \pi_1 \vee \dots \vee \pi_k$.  This shows that $(f,(\pi,\rho,\lambda))$ is the join of rank 1 elements below it, so that the map in \eqref{mapv} is injective.
\end{proof}

Let $\bar \ppdel_n$ denote the poset $\ppdel_n$ with its bottom element removed.

\begin{prop}
  The geometric realizations of the two simplicial complexes $\Omega(\bar \ppdel_n)$ and $\Omega(\bar \pp_n)$ are homotopy equivalent.  Moreover, $\tilde H_{n-2}( \bar \ppdel_n  )$ and $\tilde H_{n-2}(\bar \pp_n)$ are equivalent as representations of $\mathfrak{S}_n$.
\end{prop}

\begin{proof}
This is a direct application of Quillen's fiber poset theorem (see~\cite[Theorem~5.2.1]{wachs}), similar to the argument in \cite{athanasiadistzanaki} (as described at the end of Section~\ref{sec:assoc}).  The statement about homology follows from the equivariant version (see~\cite[Theorem~5.2.2]{wachs}). Consider the projection $ \bar \ppdel_n \to \bar\pp_n$ defined by $(f,\phi) \mapsto \phi$.  To apply the fiber poset theorem, we need to check that the fibers
\begin{equation} \label{eq:fiber}
  \big\{ (f',\phi') \in \ppdel_n \;\big|\; \phi' \leq \phi \big\}
\end{equation}
are topologically trivial for any $\phi \in \pp_n $.  

Let $(f',\phi')$ be in the set~\eqref{eq:fiber}.  Using Lemma~\ref{lem:unique}, we see that $\phi'\in\pp_n$ is uniquely determined from $f'$ and $\phi$.  So the fiber in Equation~\eqref{eq:fiber} is isomorphic to its projection to the first factor $\Delta_n$ (we have seen in the proof of the previous proposition that this projection respects the order).  If $\phi=(\pi,\rho,\sigma)$, the image of this projection is the subcomplex
\begin{equation} \label{eq:fiber2}
  \big\{ f \in \Delta_n \;\big|\; K(\underline{f}) \leq \pi \big\} =   \big\{ f \in \Delta_n \;\big|\; \underline{f} \geq K^{-1}(\pi) \big\}.
\end{equation}
This is easily seen to be isomorphic to the product $\Delta_{n_1}\times\Delta_{n_2}\times \cdots$ where $n_1,n_2,\dots$ are the block sizes of $K^{-1}(\pi)$, so it is topologically trivial since each factor $\Delta_{n_i}$ is.  So the fiber in \eqref{eq:fiber} is topologically trivial as well, and we can apply the poset fiber theorem.
\end{proof}

The geometric realizations of the two simplicial complexes $\bar \ppdel_n$ and $\Omega(\bar\ppdel_n)$ are homeomorphic (they are related by barycentric subdivision, as explained at the end of Section~\ref{sec:assoc}).  Together with the previous proposition, this shows that $\ppdel_n$ as a simplicial complex has the same topology as $\Omega(\pp_n)$.  

It might be possible to use this property in order to get an alternative way to obtain the homotopy type of $\pp_n$.  Proving shellability of $\ppdel_n$ only requires to find an appropriate total order on its maximal elements, which is potentially easier than ordering maximal chains of $\pp_n$.

\section{\texorpdfstring{$k$-divisible noncrossing $2$-partitions}{k-divisible noncrossing 2-partitions}}

\subsection{\texorpdfstring{$k$-divisible noncrossing partitions}{k-divisible noncrossing partitions}}

Let $k\geq 1$ be an integer.  The poset of $k$-divisible noncrossing partitions (of size $n$) was introduced by Edelman~\cite{edelman}, as the (full) subposet of $NC_{kn}$ containing elements $\pi$ such that the cardinality of each block is a multiple of $k$.  We use an equivalent formulation, due to Armstrong~\cite[Chapter~3]{armstrong} in a more general context.  It relies on the embedding of $NC_n$ in $\mathfrak{S}_n$ and identifies $k$-divisible noncrossing partitions with $k$-element chains in $NC_n$.  The equivalence between Edelman's definition and Armstrong's definition is stated in~\cite[Section~4.3]{armstrong}.

\begin{defi}
If $\pi\leq\tau$ in $NC_n$, define their {\it relative Kreweras complement} $K(\pi,\tau)\in NC_n$ as the unique $\nu\in NC_n$ such that $\overline{\nu} = \overline{\pi}\,\overline{\tau}^{-1}$.
The poset $NC_n^{(k)}$ of $k$-divisible noncrossing partitions is such that:
\begin{itemize}
    \item an element of $NC_n^{(k)}$ is a $k$-element chain $\pi_1\leq\dots\leq\pi_k$ in $NC_n$,
    \item There is a relation $(\pi_1,\dots,\pi_k) \leq (\tau_1,\dots,\tau_k)$ in $NC_n^{(k)}$ iff we have $\forall i\in\llbracket 1;k\rrbracket$, $K(\pi_{i-1},\pi_i) \allowbreak \geq K(\tau_{i-1},\tau_i)$ (with the convention $\pi_0=\tau_0=0_n$).
\end{itemize}
\end{defi}

Note that it might seem unnatural to have $\geq$ rather than $\leq$ in the last condition above.  We could change the definition of the relative Kreweras complement (by composing with the Kreweras complement) to have the inequality in the other way around.  If $\pi\leq\tau$, then $K(\pi,\tau)=1_n$ if and only if $\pi=\tau$, and roughly speaking $\pi$ and $\tau$ are close to each other if $K(\pi,\tau)$ is close to $1_n$.

Let us mention some properties taken from~\cite{armstrong}.  Because we took different conventions, our poset $NC_n^{(k)}$ is isomorphic to the one denoted ``$NC_{(k)}(A_{n-1})$'' in \cite{armstrong}.

\begin{prop}[{\cite[Theorem~3.4.4]{armstrong}}] \label{nckrank}
The poset $NC_n^{(k)}$ is ranked, with rank function given by 
\[
  (\pi_1,\dots,\pi_k) \mapsto |\pi_k|-1.
\]
In particular, the poset $NC_n^{(k)}$ has one minimal element, namely $(0_n,\dots,0_n)$, and its maximal elements are the $(\pi_1,\dots,\pi_k)\in NC_n^{(k)}$ such that $\pi_k=1_n$.  
\end{prop}

\begin{prop}[{\cite[Theorem~3.6.7]{armstrong}}] \label{prop:nck_chains}
  There is a bijection between $j$-element chains in $NC_n^{(k)}$ and $jk$-element chains in $NC_n$.
\end{prop}

By a result of Armstrong and Thomas (see~\cite{armstrong}), the poset $NC_n^{(k)}$ has nice topological properties.  As in the case of $\pp_n$, we denote $\hat {NC}_n^{(k)}$ the (bounded) poset obtained from $NC_n^{(k)}$ by adding a top element $\hat 1$.

\begin{prop}[{\cite[Theorem~3.7.2]{armstrong}}]
The poset $\hat{NC}_n^{(k)}$ is shellable.
\end{prop}

It follows that $NC_n^{(k)}$ is a Cohen-Macaulay poset.  

\begin{rema}
The bijection in Lemma~\ref{lemma_encodingNC} can be extended to the case of $k$-divisible noncrossing partitions, at the condition of using Edelman's definition.  Indeed, it clearly gives a bijection between elements $\pi\in NC_{kn}$ having block sizes divisible by $k$, and weak compositions $(a_1,\dots,a_{kn})$ of $kn$ such that each $a_i$ is a multiple of $k$ and $\sum_{i=1}^j a_i \geq j$, $\forall j \in \llbracket 1;kn\rrbracket$.  These are obviously in bijection with weak compositions $(a_1,\dots,a_{kn})$ of $n$ such that $\sum_{i=1}^j ka_i \geq j$, $\forall j \in \llbracket 1;kn\rrbracket$.
\end{rema}

\subsection{\texorpdfstring{$k$-divisible noncrossing $2$-partitions}{k-divisible noncrossing 2-partitions}}

Just as in the case of noncrossing partitions, the idea is to define here an order relation on $k$-element chains of $\pp_n$.

\begin{defi} \label{defppnk}
 The poset $\pp_n^{(k)}$ of {\it $k$-divisible noncrossing $2$-partitions} is defined as the set of $k$-element chains of $\pp_n$, with the order relation such that:
\[
  ((\pi_1,\rho_1,\lambda_1),\dots,(\pi_k,\rho_k,\lambda_k))
  \leq 
  ((\pi'_1,\rho'_1,\lambda'_1),\dots,(\pi'_k,\rho'_k,\lambda'_k)
\]
iff $(\pi_1,\dots,\pi_k) \leq (\pi'_1,\dots,\pi'_k)$ in $NC_n^{(k)}$ and $(\pi_k,\rho_k,\lambda_k) \leq (\pi'_k,\rho'_k,\lambda'_k)$ in $\pp_n$.
\end{defi}

Note that a chain $((\pi_i,\rho_i,\lambda_i))_{1\leq i \leq k}$ as above can be concisely encoded in the tuple $(\pi_1,\dots, \pi_k,\rho_k,\lambda_k)$, using Lemma~\ref{lem:unique}.  We can see $\pp_n^{(k)}$ as the fiber product of $NC_n^{(k)}$ and $\pp_n$ over $NC_n$, along the two posets maps $\pp_n^{(k)} \to NC_n^{(k)}$,
$(\pi_1,\dots, \pi_k,\rho_k,\lambda_k) \mapsto (\pi_1,\dots, \pi_k)$ and $\pp_n^{(n)} \to \pp_n$, $(\pi_1,\dots, \pi_k,\rho_k,\lambda_k) \mapsto (\pi_k, \rho_k , \lambda_k)$.

In analogy with Lemma~\ref{lem:unique}, we have:

\begin{lemm}\label{lem:uniquek}
  Let $(\pi_1,\dots, \pi_k,\rho_k,\lambda_k) \in \pp_n^{(k)}$.  Then, for each  $(\pi'_1,\dots, \pi'_k)$ below $(\pi_1,\dots, \pi_k)$ in $NC_n^{(k)}$, there exist unique $\rho'_k$ and $\lambda'_k$ such that $(\pi'_1,\dots, \pi'_k,\rho'_k,\lambda'_k)$ is below $(\pi_1,\dots, \pi_k,\rho_k,\lambda_k)$ in $\pp_n^{(k)}$.  
\end{lemm}

\begin{proof}
  This is a direct consequence of the definition of $\pp_n^{(k)}$, together with Lemma~\ref{lem:unique}.
\end{proof}

It follows from the previous lemma that there is an isomorphism between the order ideal of $NC_n^{(k)}$ containing elements below $(\pi_1,\dots, \pi_k)$, and the order ideal of $\pp_n^{(k)}$ containing elements below $(\pi_1,\dots, \pi_k,\rho_k,\lambda_k)$.

We have seen that the action of $\mathfrak{S}_n$ preserves the order relation of $\pp_n$, thus defining an action on $\pp_n^{(k)}$ having $\operatorname{Park}_n^{(k)}$ as its character.  Moreover, it is straightforward from Definition~\ref{defppnk} that this action preserves the order relation on $\pp_n^{(k)}$.  Going one step further, we have:

\begin{prop} \label{zeta_characterk}
  There is a $\mathfrak{S}_n$-equivariant bijection between $j$-element chains of $\pp_n^{(k)}$ and $\pp_n^{(jk)}$.  In particular, the character of $\mathfrak{S}_n$ acting on $j$-element chains of $\pp_n^{(k)}$ is $\operatorname{Park}_n^{(jk)}$.
\end{prop}

\begin{proof}
From Lemma~\ref{lem:uniquek}, a $j$-element chain $\psi_1,\dots,\psi_j$ of $\pp_n^{(k)}$ can be encoded in a $j$-element chain $\pi_1,\dots,\pi_j$ in $NC_n^{(k)}$ together with the element $\psi_j = (\phi_1,\dots,\phi_k)$ of $\pp_n^{(k)}$.  Using Lemma~\ref{lem:unique}, $(\phi_1,\dots,\phi_k)$ can be recovered from $\pi_j$ and $\phi_k$.  The bijection in Proposition~\ref{prop:nck_chains} sends $\pi_1,\dots,\pi_j$ to a $jk$-element chain $\tau_1,\dots,\tau_{jk}$.
Together with $\phi_k$, this chains defines an element of $\pp_n^{(jk)}$ and it can be checked that this is an equivariant bijection.
\end{proof}

As usual, let $\hat \pp_n^{(k)}$ denote the poset $\pp_n^{(k)}$ with an extra top element $\hat 1$, and $\bar \pp_n^{(k)}$ denote $\pp_n^{(k)}$ with its bottom element removed. 

\begin{prop}
  The poset $\pp_n^{(k)}$ is Cohen-Macaulay.
\end{prop}

\begin{proof}
  We use the projection $\pp_n^{(k)} \to \pp_n$ sending a $k$-element chain to its top element.  Note that this map preserves rank by Proposition~\ref{nckrank}, and it can be checked to be increasing. The fiber associated to $(\pi,\rho,\lambda)\in\pp_n$ is the subposet of $\pp_n^{(k)}$ containing $k$-element chains with top element below $(\pi,\rho,\lambda)$.  Using Lemma~\ref{lem:uniquek}, we find that this fiber is isomorphic to the subposet
  \begin{equation}\label{fiberppk}
    \big\{ 
        (\pi_1,\dots,\pi_k) \in NC_n^{(k)} \; \big| \; \pi_k \leq \pi 
    \big\}.
  \end{equation}
  Note that this is not an order ideal in $NC_n^{(k)}$.  Let $i_1,i_2,\dots$ denote the block sizes of $K(\pi)$, so that the order ideal $\{\pi'\in NC_n\;|\;\pi'\leq\pi\}$ is isomorphic to $NC_{i_1}\times NC_{i_2}\times \cdots$.  It can be checked that this isomorphism carries over to $k$-element chains and shows that the subposet in~\eqref{fiberppk} is isomorphic to $NC^{(k)}_{i_1}\times NC^{(k)}_{i_2}\times \cdots$.

  The fiber in~\eqref{fiberppk} is Cohen-Macaulay, since it is isomorphic to a product of Cohen-Macaulay posets.  We can thus use the poset fiber theorem for Cohen-Macaulay posets (see~\cite[Theorem~5.2]{baclawski}). 
\end{proof}

Note that this method does not provide a shelling of the poset $\pp_n^{(k)}$, and it would be interesting to build one explicitly.  But the previous proposition suffices to guarantee that the only nonzero homology group of $\pp_n^{(k)}$ is the $n-2$nd.  Now, the content of Section~\ref{sec_homology} can be adapted to the $k$-divisible case.

\begin{thm} \label{theo_charhomologyk}
The character of $\tilde H_{n-2}(\bar\pp_n^{(k)})$ as a representation of $\mathfrak{S}_n$ is given by:
\begin{align}  \label{formula_charhomologyk}
  \sigma \mapsto (-1)^{n-z(\sigma)} (kn-1)^{z(\sigma)-1}.
\end{align}
\end{thm}

\begin{proof}
  As in the proof of Theorem~\ref{theo_charhomology}, we can use the zeta character of $\pp_n^{(k)}$, which is $\operatorname{Park}_n^{(jk)}$ by Proposition~\ref{zeta_characterk}.  By plugging $j=-1$ and multiplying by $(-1)^{n-1}$, we get the result.
\end{proof}

The explicit formulas for Whitney is left to the interested reader, and we finish this section with the analog of Section~\ref{sec:prime}.

\begin{defi}
  An element $(\pi_1,\dots,\pi_k) \in NC_n^{(k)}$ is {\it prime} if $\pi_k\in NC'_n$.    An element $(\phi_1,\dots,\phi_k) \in \pp_n^{(k)}$ is {\it prime} if $\phi_k\in \pp'_n$.  We denote $NC'{}_n^{(k)} \subset NC_n^{(k)}$ the subset of prime $k$-divisible noncrossing partitions, and by $\pp'{}_n^{(k)} \subset \pp_n^{(k)}$ the subset of prime $k$-divisible noncrossing $2$-partitions.
\end{defi}

Following the argument in Sections~\ref{sec_parkingspace} and~\ref{sec:prime}, the character of $\mathfrak{S}_n$ acting on $\pp'{}_n^{(k)}$ is 
\[
  \operatorname{Park'}_n^{(k)} 
  := 
  \sum_{(\pi_1,\dots\pi_k)\in NC'{}_n^{(k)}} \operatorname{Ind}_{\mathfrak{S}_n(\pi_k)}^{\mathfrak{S}_n}(1).
\]

\begin{prop} \label{prop_parkk}
  For $\sigma\in\mathfrak{S}_n$, we have:
  \[
    \operatorname{Park'}_n^{(k)}(\sigma)
    =
    (kn-1)^{z(\sigma)-1}.
  \]
\end{prop}

\begin{proof}
  Using the bijection between $\pp_n^{(k)}$ and $k$-parking functions, we can check that $\pp'{}_n^{(k)}$ is sent to $k$-parking functions $w_1,\dots,w_n$ such that $\forall j\in\llbracket1;n-1\rrbracket$, $\#\{ \, i \, : \, w_i \leq k(j-1)+1 \} > j$.  As in the case of Proposition~\ref{prop_park}, they are in bijection with the rational parking functions of Armstrong, Loehr and Warrington~\cite{ALW}, this time with the parameters $(a,b)=(n,kn-1)$.  So the result follows from the formula for the character of $(a,b)$-parking functions.  We omit details.
\end{proof}

From Theorem~\ref{theo_charhomologyk} and Proposition~\ref{prop_parkk}, we get the following:

\begin{coro}
  The character of $\tilde H_{n-2}(\bar\pp_n^{(k)})$ is 
  \[
    \operatorname{Char}(\tilde H_{n-2}(\bar \pp_n^{(k)})) = \operatorname{Sign} \otimes \operatorname{Park'}_n^{(k)}.
  \]
\end{coro}

As explained in Section~\ref{sec:prime} for the case $k=1$, it would be interesting to prove this corollary by finding an explicit basis $(e_\phi)_{\phi\in \ppp_n'^{(k)}}$ of $\tilde H_{n-2}(\bar\pp_n^{(k)})$ such that $\sigma\cdot e_{\phi} = \operatorname{Sign}(\sigma) e_{\sigma\cdot\phi}$.

\subsection{$k$-divisible noncrossing partitions in the sense of Edelman}

Let us record Edelman's definition:

\begin{defi}[\cite{edelman}]
  We define $NC^{[k]}_n$ as the subposet of $NC_{kn}$ containing elements all of whose blocks have cardinality divisible by $k$.
\end{defi}

There is a bijection between $NC_n^{[k]}$ and $k$-trees with $kn+1$ vertices, obtained as a restriction of the bijection $\beta$ from Section~\ref{sec:nc}.  It is natural to introduce the following:

\begin{defi}
  We define $\pp^{[k]}_n$ as the subposet of $\pp_{kn}$ containing elements all of whose blocks have cardinality divisible by $k$.
\end{defi}

The action of $\mathfrak{S}_{kn}$ on $\pp_{kn}$ stabilizes $\pp^{[k]}_n$, and the orbits are naturally indexed by $NC_n^{[k]}$.  For the next proposition, recall that $\chi$ denote the Frobenius characteristic map, which sends a $\mathfrak{S}_n$-character (equivalently, a $\mathfrak{S}_n$-set) to a symmetric function of degree $n$.  

\begin{prop}
  The symmetric function $\chi(\pp_n^{[k]})$ is the image of $\chi(\pp_n^{(k)})$ by the algebra morphism that sends the $i$th homogeneous symmetric function $h_i$ to $h_{ki}$.
\end{prop}

\begin{proof}
  First note that we have:
  \[
    \chi(\pp_n^{[k]})
      =
    \sum_{T} \prod_{v\in T} h_{\deg(v)}
  \]
  where we sum over $k$-trees with $kn+1$ vertices, the product is over internal vertices $v$ of $T$, and $\deg(v)$ is the number of descendants of $v$.  This follows from Proposition~\ref{prop:2pp_park}, using trees rather than noncrossing partitions (via the bijection $\beta$).  So, it remains to show:
  \[
    \chi(\pp_n^{(k)})
    =
    \sum_{T} \prod_{v\in T} h_{\deg(v)/k}
  \]
  This is clear from the interpretation in terms of $k$-parking trees.
\end{proof}

\section{Perspectives}

Let us just mention some further questions arising from this work.  First, as said in Remark~\ref{rem:shellabilitycriterion}, it would be interesting to investigate if our criterion for shellability is equivalent to CL-shellability or if there exists a poset satisfying our criterion but not CL-shellable.  It would be moreover very interesting to find posets where our criterion is particularly suited (besides $\pp_n$).

Also, there should be a generalisation of Edelman's poset to other finite reflection groups. Indeed, in this context there is an associated noncrossing partition lattice, and a noncrossing parking space.  New methods might be needed to prove shellability in this general setting.

It would be very interesting to have a bijective proof of Proposition~\ref{prop:nck_chains} in terms of planar trees (avoiding the relative Kreweras complement), which could be adapted to parking trees to get Proposition~\ref{zeta_characterk}.  The first step towards this proof is to find an appropriate order on the tree representation of $NC_n^{(k)}$.


\begin{thebibliography}{xx}
\label{sec:bibli}

\bibitem{armstrong}
  \textsc{D. Armstrong}: 
  ``Generalized noncrossing partitions and combinatorics of Coxeter groups''.
  {\it Mem. Amer. Math. Soc.} 202, no. 949. 2009

\bibitem{ARR}
  \textsc{D. Armstrong, V. Reiner, B. Rhoades}:
  ``Parking spaces''. 
  {\it Adv. Math.} 269 (2015), 647--706.

\bibitem{ALW}
  \textsc{D. Armstrong, N. A. Loehr, G. S. Warrington }:
  ``Rational Parking Functions and Catalan Numbers''.
  {\it Ann. Comb.} 20 (2016),  21--58.

\bibitem{athanasiadis}
  \textsc{C. A. Athanasiadis}:
  ``On some enumerative aspects of generalized
associahedra''.
  {\it Europ. J. of Combin.} 28 (2007) 1208--1215

\bibitem{athanasiadistzanaki}
  \textsc{C. A. Athanasiadis, E. Tzanaki}:
  ``Shellability and higher Cohen-Macaulay connectivity of generalized cluster complexes''.
  Israel J. Math. 167, (2008), Article number: 177.
  
\bibitem{baclawski} 
  \textsc{K. Baclawski}:
  ``Cohen-Macaulay ordered sets''.
  J. Algebra 63 (1980), 226--258.
 
\bibitem{Biane}
  \textsc{P. Biane}:
  ``Some properties of crossings and partitions''. 
  {\it Disc. Math.}, 175 (1--3) (1997), 41--53.

\bibitem{bjorner80}
 \textsc{A. Björner}: 
 ``Shellable and Cohen-Macaulay partially ordered sets''. 
 {\it Trans. Amer. Math. Soc.} 260 (1980), 159--183.

\bibitem{bjornerwachs83}
 \textsc{A. Björner, M.L. Wachs},
 ``On lexicographically shellable posets''. 
 {\it Trans. Amer. Math. Soc.} 277 (1983), 323--341.

\bibitem{BLL}
  \textsc{F. Bergeron, G. Labelle, P. Leroux}:
  ``Combinatorial species and tree-like structures''.
  Cambridge University Press, Cambridge, 1998.

\bibitem{sppp}
  \textsc{B. Delcroix-Oger}:
  ``Semi-pointed partition posets and species''.
  {\it J. of Alg. Comb.} 45 (2017), pp 857--886.

\bibitem{fpsac2020}
  \textsc{B. Delcroix-Oger, M. Josuat-Vergès, L. Randazzo}: ``Some properties of the parking function poset'' (extended abstract).  Sém. Loth. Comb. 84B (2020), Article \#59, 12 pp.
  Proceedings 32nd FPSAC.

\bibitem{edelman}
  \textsc{P. Edelman}:
  ``Chain enumeration and noncrossing partitions''.
  {\it Disc. Math.}, 31(2) (1980), 171--180.

\bibitem{haiman}
  \textsc{M. Haiman}: 
  ``Conjectures on the quotient ring by diagonal invariants''.
  {\it J. of Alg. Comb.} 3 (1994), 17--76.

 
\bibitem{kozlov2}
  \textsc{D. N. Kozlov}
  ``Spectral sequences on combinatorial simplicial complexes''.
  J. of Alg. Comb. 14 (2001), 27--48.

\bibitem{kreweras}
  \textsc{G. Kreweras}:
  ``Sur les partitions non croisées d'un cycle''.
  {\it Disc. Math.} 1(4) (1972), 333--350.

\bibitem{laradjiumar}
  \textsc{A. Laradji, A. Umar}: 
  ``On the number of nilpotents in the partial symmetric semigroup''. 
  {\it Comm. Alg.}, 32(8) (2004), 3017--3023. 
 
\bibitem{macdonald}
 \textsc{I. G. Macdonald}.
 Symmetric functions and Hall polynomials (2nd edition).  Clarendon Press, Oxford, 1995, x + 475 pp.
 
\bibitem{tamari}
  \textsc{F. Müller-Hoissen, J. M. Pallo, J. Stasheff}, editors:
  ``Associahedra, Tamari lattices and related structures. Tamari memorial Festschrift.'' Progress in Mathematical Physics, 299. Birkhäuser/Springer, Basel, 2012. xx+433 pp.
 
\bibitem{munkres}
 \textsc{J.R. Munkres} 
 ``Elements of Algebraic Topology''.
 Perseus Publishing, 1984.
 
\bibitem{rhoades}
 \textsc{B. Rhoades}:
 ``Parking structures: Fuss analogs''. 
 {\it J.  Alg. Combin.} 40 (2014), pp.~417--473.

\bibitem{wachs} 
  \textsc{M. Wachs}:
  ``Poset topology: tools and applications''. 
  {\it Geometric Combinatorics}, IAS/Park City Mathematics Series, AMS, Providence, 2007, pp.~497--615.
 
\bibitem{yan} 
  \textsc{C. Yan}:
  ``Generalized parking functions, tree inversions, and multicolored graphs''. 
  {\it Adv. in Applied Math.} 27(2--3) (2001), pp.~641--670.

\end{thebibliography}
\end{document}